%% file: Multicolored_Coverage_long.tex
\documentclass[a4paper,11pt]{article}
\usepackage[margin=1in]{geometry}

\usepackage{tabularx}
\usepackage{epsfig}
\usepackage{xfrac}
\newcommand{\vccc}{\textsc{Max $k$-VC}\xspace}
\newcommand{\kddfree}{{$K_{d,d}$-free}\xspace}
\newcommand{\cR}{\mathcal{R}}

\newcommand{\hidestuff}[1]{}
\newcommand{\cl}{\mathsf{cl}}
\newcommand{\rank}{\mathsf{rank}}

\usepackage{hyperref}
\hypersetup{
    colorlinks=true, %set true if you want colored links
    linktoc=all,     %set to all if you want both sections and subsections linked
    linkcolor=blue,
    citecolor=red,  %choose some color if you want links to stand out
}
\usepackage{tcolorbox}
\usepackage{tikz}

\title{Satisfiability to Coverage in Presence of Fairness, Matroid, and Global Constraints}

\author{
	Tanmay Inamdar\thanks{
		Indian Institute of Technology Jodhpur, Jodhpur, India.}
	\and
	Pallavi Jain\addtocounter{footnote}{-1}\footnotemark{}
	\and
	Daniel Lokshtanov\footnote{University of California Santa Barbara, United States}
	\and
	Abhishek Sahu\footnote{National Institute of Science Education and Research (NISER), Bhubaneswar, India.}
	\and
	Saket Saurabh\footnote{Insitute of Mathematical Sciences, Chennai, India, and University of Bergen, Bergen, Norway.}
	\and
	Anannya Upasana\footnote{Institute of Mathematical Sciences, Chennai, India.
	\\S.~Saurabh acknowledges support from the European Research Council (ERC) under the European Union’s Horizon 2020 research and innovation programme (grant agreement No. 819416) Swarnajayanti Fellowship (No. DST/SJF/MSA01/2017-18). T.~Inamdar acknowledges support from the European Research Council (ERC) under the European Union’s Horizon 2020 research and innovation programme (grant agreement No. 819416). P.~Jain acknowledges support from SERB-SUPRA grant number S/SERB/PJ/20220047 and IITJ Seed Grant grant I/SEED/PJ/20210119.} 
}

\date{}

\date{}

\input{preamble}

\newcommand{\OO}{{\mathcal O}}
\newcommand{\satcc}{\textsc{CC-MaxSat}\xspace}

\newcommand{\maxcov}{\textup{\textsc{Maximum Coverage}}\xspace}

\newcommand{\psc}{\textup{\textsc{Partial Set Cover}}\xspace}
\newcommand{\partsc}{\textup{\textsc{Partition Maximum Coverage}}\xspace}
\newcommand{\yin}{\textup{\textsf{Yes}}-instance\xspace}
\newcommand{\nin}{\textup{\textsf{No}}-instance\xspace}

\newcommand{\cM}{\mathcal{M}}

\newcommand{\rep}[1]{\displaystyle\subseteq^{#1}_{\text{rep}}}

\begin{document}
\maketitle

\begin{abstract}
In the {\sc MaxSAT} with Cardinality Constraint problem ({\sc CC-MaxSAT}), we are given a CNF-formula $\Phi$, and a positive integer $k$, and the goal is to find an assignment $\beta$ with at most $k$ variables set to true (also called a weight $k$-assignment) such that the number of clauses satisfied by $\beta$ is maximized. \maxcov  can be seen as a special case of \satcc, where the formula $\Phi$ is monotone, i.e., does not contain any negative literals. \satcc and \maxcov are extremely well-studied problems in the approximation algorithms as well as parameterized complexity literature.

Our first conceptual contribution is that \satcc and \maxcov  are equivalent to each other in the context of 
FPT-Approximation parameterized by $k$ (here, the approximation is in terms of number of clauses satisfied/elements covered). In particular, we give a randomized reduction from  \satcc to \maxcov running in time $\Oh(1/\epsilon)^{k} \cdot (m+n)^{\Oh(1)}$ that preserves the approximation guarantee up to a factor of $1-\epsilon$. Furthermore, this reduction also works in the presence of ``fairness'' constraints on the satisfied clauses, as well as matroid constraints on the set of variables that are assigned $\mathsf{true}$. Here, the ``fairness'' constraints are modeled by partitioning the clauses of the formula $\Phi$ into $r$ different colors, and the goal is to find an assignment that satisfies at least $t_j$ clauses of each color $1 \le j \le r$.

Armed with this reduction, we focus on designing FPT-Approximation schemes (\FPTAS{}es) for \maxcov and its generalizations. Our algorithms are based on a novel combination of a variety of ideas, including a carefully designed probability distribution that exploits sparse coverage functions.
%, representative families for linear matroids. 
These algorithms substantially generalize the results in Jain et al.~[SODA 2023] for  \satcc and \maxcov  for $K_{d,d}$-free set systems (i.e., no $d$ sets share $d$ elements), as well as a recent \FPTAS for {\sc Matroid Constrained Maximum Coverage} by Sellier~[ESA 2023] for frequency-$d$ set systems.
 %\todo[inline]{Define FPT-ASes, why in the reduction running time we have $r$, {\sc Matroid Constrained Maximum Coverage} -- need some short name. Uniformity on FPT-AS with sf or without conclusuion}
\end{abstract}

%\newpage

%\tableofcontents
%
%\newpage
\pagenumbering{arabic}

%\tableofcontents
%\newpage

\section{Introduction} \label{sec:intro}

%In \setcover, and its closely related variant, \maxcov, we are interested in optimizing the \emph{coverage properties} of a set system. Such coverage problems arise naturally in many applications. Thus, these two problems are arguably the most well-studied problems in combinatorial optimization. 

%\todo[inline]{add definition for matroid constraints.}
%\todo[inline]{maybe say a generic definition for linear matroid and then we first discuss results for uniform matroid.}
%\todo[inline]{add $K_{d,d}$-free result back.}

Two problems that have gained considerable attention from the perspective of Parameterized Approximation~\cite{DBLP:journals/algorithms/FeldmannSLM20} are  the classical {\sc MaxSAT} with cardinality constraint (\satcc) problem and its monotone version,   the \maxcov  problem. In the  \satcc problem, we are given a CNF-formula $\Phi$ over $m$ clauses and $n$ variables, and a positive integer $k$, and the objective is to find a weight $k$ assignment that maximizes the number of satisfied clauses. We use $\var{\Phi}$ and $\cla{\Phi}$ to denote the set of variables and clauses in $\Phi$,  respectively. An {\em assignment}  to a  CNF-formula $\Phi$ is a function $\beta :  \var{\Phi}\rightarrow \{0,1\}$. The {\em weight} of an assignment $\beta$ is the number of variables that have been assigned $1$. 

The  classical  \maxcov  problem is a special case of the \satcc problem.  Indeed, it is a monotone variant of \satcc, where negated literals are not allowed. 
An input to the \maxcov problem consists of a family of $m$ sets, $\cal F$, over a universe $U$ of size $n$, and an integer $k$, and the goal is to find a subfamily ${\cal F}' \subseteq {\cal F}$ of size $k$ such that the number of elements {\em covered} (belongs to some set in ${\cal F}'$)  by ${\cal F}'$ is maximized. Observe that when the goal is to cover every element in $U$, the \maxcov problem corresponds to {\sc Set Cover}.  
A natural question that has guided research on these problems is whether \satcc or \maxcov  admits an algorithm with running time $f(k)n^{\OO(1)}$? That is, whether \satcc or \maxcov  is fixed parameter tractable (\FPT) with solution size $k$? Unfortunately, these problems are W[2]-hard~\cite{DBLP:books/sp/CyganFKLMPPS15}. That is, we do not expect these problems to  admit an algorithm with running time $f(k)n^{\OO(1)}$. This negative result sets the platform for studying these problems from the viewpoint of  Parameterized Approximation~\cite{DBLP:journals/algorithms/FeldmannSLM20}. It is well known that both  \satcc and \maxcov admit a polynomial time $(1-\frac{1}{e})$-approximation algorithm~\cite{DBLP:journals/algorithmica/Sviridenko01}, which is in fact optimal.~\cite{DBLP:conf/stoc/Feige96}. So, in the realm of  Parameterized Approximation, we ask  does  there exist an $\epsilon >0$, such that \satcc or \maxcov admits an approximation algorithm with factor $(1-\frac{1}{e} +\epsilon)$ and runs in time 
$f(k,\epsilon)n^{\OO(1)}$. While there has been a lot of work on   \maxcov~\cite{DBLP:conf/soda/0001KPSS0U23,DBLP:journals/corr/abs-1810-03792,DBLP:journals/jair/SkowronF17,DBLP:journals/tcs/HuangS23,DBLP:conf/esa/Sellier23}, Jain et al.~\cite{DBLP:conf/soda/0001KPSS0U23} studied \satcc and designed a standalone algorithm for the problem. Our first result, a bit of a surprise to us,  shows that in the world of Parameterized Approximation \satcc and  \maxcov are {\em ``equivalent". }

\begin{theorem}[Informal]
	\label{thmintro:reduction}
	Let $\epsilon >0$. There is a polynomial time randomized algorithm that given an instance $(\Phi, k)$  of \satcc   produces an instance $(U, {\cal F}, k)$ of  \maxcov such that the following holds with probability $\frac{1}{2}(\frac{\epsilon}{2})^{k}$. Given a 
	$(1-\epsilon) {\sf OPT_{cov}}$ solution to $(U, {\cal F}, k)$ we can obtain a $(1-\epsilon) {\sf OPT_{sat}}$ solution to 
	$(\Phi, k)$ in polynomial time. Here, ${\sf OPT_{cov}}$ (${\sf OPT_{sat}}$) denotes the value of the maximum number of covered elements (satisfied clauses) by a $k$ sized family of subsets (weight $k$ assignment).  
\end{theorem}

Theorem~\ref{thmintro:reduction} allows us to focus on \maxcov, rather than \satcc, at the expense of $\epsilon^{-\Oh(k)}$ in the running time. Further, there is no assumption on the input formulas in Theorem~\ref{thmintro:reduction}.  This reduction immediately implies faster algorithms for \satcc by utilizing the known good algorithms for \maxcov~\cite{DBLP:conf/soda/0001KPSS0U23,DBLP:journals/corr/abs-1810-03792,DBLP:journals/jair/SkowronF17,DBLP:journals/tcs/HuangS23,DBLP:conf/esa/Sellier23}. The \maxcov problem has been generalized in several directions by adding either fairness constraints or asking our solution to be an independent set of a matroid. In what follows, we take a closer look at progresses on \maxcov and its generalizations and then design algorithms that generalize and unify all the known results for \satcc and \maxcov.  

\subsection{Tractability Boundaries for \maxcov}
Cohen-Addad et al.~\cite{DBLP:conf/icalp/Cohen-AddadG0LL19} studied \maxcov and showed that there is no $\epsilon >0$, such that \maxcov admits an approximation algorithm with factor $(1-\frac{1}{e} +\epsilon)$ and runs in time $f(k,\epsilon) (m+n)^{\OO(1)}$ \footnote{Throughout the paper, the approximation factor will refer to the number of elements covered/number of satisfied clauses, unless explicitly stated otherwise}. Later, this was also studied by Manurangsi~\cite{DBLP:journals/corr/abs-1810-03792}, who obtained the following  strengthening over~\cite{DBLP:conf/icalp/Cohen-AddadG0LL19}: for any constant $\epsilon >0$ and any function $h$, assuming Gap-ETH, no $h(k)(n+m)^{o(k)}$ time algorithm can approximate \maxcov with $n$ elements and $m$ sets to within a factor $(1-\frac{1}{e} +\epsilon)$, even  with  a  promise  that  there 
exist $k$ sets that fully cover the whole universe.  This negative result sets the contour for possible positive results. 
In particular, if we hope for an \FPT algorithm that improves over a factor $(1-\frac{1}{e})$ then we must assume some additional structure on the input families. This automatically leads to the families wherein each set has bounded size, or each element appears in bounded sets which was considered earlier. 

Skowron and  Faliszewski~\cite{DBLP:journals/jair/SkowronF17}  showed that, if we are working on set families, such that each element in $U$ appears in at most $p$ sets, then there exists an algorithm, that given an $\epsilon >0$, runs in time $(\frac{p}{\epsilon})^{\OO(k)}n^{\OO(1)}$ and returns a subfamily ${\cal F}'$ of size $k$ that is a $(1-\epsilon)$-approximation. These kind of \FPT-approximation algorithms are called \FPT-approximation Schemes (\FPTAS{}es). For $p=2$, Manurangsi \cite{DBLP:journals/corr/abs-1810-03792} independently obtained a similar result.
%For the case of $d=2$, the current fastest algorithm was independently given by  Manurangsi~\cite{DBLP:journals/corr/abs-1810-03792} and Skowron and  Faliszewski~\cite{DBLP:journals/jair/SkowronF17}, and runs in time $(\frac{1}{\epsilon})^{\OO(k)}n^{\OO(1)}$.  
Jain et al.~\cite{DBLP:conf/soda/0001KPSS0U23}  generalized these two settings by looking at $K_{d,d}$-free set systems (i.e., no $d$ sets share $d$ elements).  They also considered  $K_{d,d}$-free formulas (that is, the clause-variable incidence bipartite graph of the formula excludes $K_{d,d}$ as an induced subgraph). They showed that for every $\epsilon>0$, there exists an algorithm for $K_{d,d}$-free formulas with approximation ratio $(1-\epsilon)$ and running in time  $2^{\mathcal{O}((\frac{dk}{\epsilon})^d)}  (n+m)^{\mathcal{O}(1)}$. For,  \maxcov on  $K_{d,d}$-free set families, they obtain an \FPTAS with running time $(\frac{dk}{\epsilon})^{\mathcal{O}(dk)}n^{\mathcal{O}(1)}$. Using these results together with Theorem~\ref{thmintro:reduction} we get the following. 

\begin{corollary}
	\label{introcor:satcc}
	Let $\epsilon >0$. Then, \satcc admits a randomized  \FPTAS with running time $(\frac{dk}{\epsilon})^{\mathcal{O}(dk)}n^{\mathcal{O}(1)}$ on $K_{d,d}$-free formulas. Furthermore, if the size of clauses is bounded by $p$ or every variable appears in at most $p$ clauses then \satcc admits  randomized \FPTAS with running time $(\frac{p}{\epsilon})^{\OO(k)}n^{\OO(1)}$. Both results hold with constant probability. 
\end{corollary}

Corollary~\ref{introcor:satcc} follows by utilizing Theorem~\ref{thmintro:reduction} and repurposing the known results about \maxcov (\cite{DBLP:conf/soda/0001KPSS0U23,DBLP:journals/ipl/Blaser03,DBLP:journals/jair/SkowronF17,DBLP:journals/corr/abs-1810-03792}). We will return to the case of $K_{d,d}$-free set systems later. Apart from extending the classes of set families where \maxcov admits   \FPTAS{}es, the study on the \maxcov problem has been extended in many directions.

\subsubsection{Matroid Constraints}
Note that \maxcov is a special case of submodular function maximization subject to a cardinality constraint. In the latter problem, we are given (an oracle access to) a submodular function $f: 2^{V} \to \mathbb{R}_{\ge 0}$ \footnote{$f: 2^V \to \mathbb{R}$ is submodular if it satisfies $f(A) + f(B) \ge f(A \cup B) + f(A \cap B)$ for all $A, B \subseteq V$}, and the goal is to find a subset $U \subseteq V$ that maximizes $f(U)$ over all subsets of size at most $k$. Indeed, coverage functions are submodular and monotone (i.e., adding more sets cannot decrease the number of elements covered). There has been a plethora of work on monotone submodular maximization subject to cardinality constraints, starting from Wolsey \cite{DBLP:journals/combinatorica/Wolsey82}. In a further generalization, we are interested in monotone submodular maximization subject to a matroid constraint -- in this setting, we are given a matroid $\cM = (U, \cI)$ \footnote{Recall that a matroid is a pair $\cM = (U, \cI)$, where $U$ is the ground set, and $\cI$ is a family of subsets of $U$ satisfying the following three axioms: (i) $\emptyset \in \cI$, (ii) If $A \in \cI$, then $B \in \cI$ for all subsets $B \subseteq A$, and (iii) for any $A, B \in \cI$ with $|B| > |A|$, then there exists an element $e \in B \setminus A$ such that $A \cup \LR{e} \in \cI$.} via an \emph{independence oracle}, i.e., an algorithm that answers queries of the form ``Is $P \in \cI$?'' for any $P \subseteq U$ in one step, and we want to find an independent set $S \in I$ that maximizes $f(S)$. Note here that a uniform matroid of rank $k$ \footnote{Rank of a matroid is equal to the maximum size of any independent set in the matroid.} exactly captures the cardinality constraint. Calinescu et al. \cite{CalinescuCPV11} gave an optimal $(1-1/e)$-approximation. 

More recently, Huang and Sellier \cite{DBLP:journals/tcs/HuangS23} and Sellier \cite{DBLP:conf/esa/Sellier23} studied the problem of maximizing a coverage function subject to a matroid constraint, called \textsc{Matroid Constrained Maximum Coverage}. In this problem, which we call \mmaxcov (M for ``matroid'' constraint), we are given a set system $(U, \cF)$ and a matroid $\cM = (\cF, I)$ of rank $k$, and the goal is to find a subset $\cF' \subseteq \cF$ such that $\cF' \in I$ and $\cF'$ maximizes the number of elements covered. Note that \mmaxcov is a generalization of \maxcov. In the latter paper, Sellier \cite{DBLP:conf/esa/Sellier23} designed an \FPTAS for \mmaxcov, running in time $(d/\epsilon)^{\Oh(k)} \cdot (m+n)^{\Oh(1)}$ for frequency-$d$ set systems. Note that this result generalizes that of \cite{DBLP:journals/jair/SkowronF17, DBLP:journals/corr/abs-1810-03792} from a uniform matroid consraint to an arbitrary matroid constraint of rank $k$. 

Analogous to \mmaxcov, one can define a matroid constrained version of \satcc, called \mmaxsat. In this problem, we are given a CNF-SAT formula $\Phi$ and a matroid $\cM$ of rank $k$ on the set of variables. The goal is to find an assignment that satisfies the maximum number of clauses, with the restriction that, the set of variables assigned $1$ must be an independent set in $\cM$. Note that \mmaxsat generalizes \mmaxcov as well as \satcc. We obtain the following result for \mmaxsat, by combining the results on a variant of \Cref{thmintro:reduction} with the corresponding result on \mmaxcov. 

\begin{theorem} \label{thm:intro-matroid-sat}
	There exists an \FPTAS for \mmaxsat parameterized by $k, d$, and $\epsilon$, on $d$-CNF formulas, where $k$ denotes the rank of the given matroid.
\end{theorem}

\subsubsection{Fairness or Multiple Coverage Constraints}\label{subsec:fairness}%\todo{Partiton or fairness?}

%\maxcov has a natural application in multiwinner election model, where, one wants to satisfy as many voters as possible, by picking candidates of their choice. In particular, the well-known Chamberlin-Courant rule in approval model is nothing but \maxcov.  Here, we can imagine that a voter is an element of the universe and a candidate corresponds to a set in the set system.  However, this may underrepresent a segment of voters, and thus naturally, the ``concept of fairness'' comes into play.
%
%%In the last few years, the \emph{fair} versions of various classical combinatorial optimization problems have been steadily gaining popularity. There are numerous definitions of fairness that are applicable in different scenarios. 
%%\todo{A few lines motivating fairness}
%
%One of the models of fairness is that voters are divided into groups based on gender, region, community, etc., and the objective is to select a group of candidates such that they maximise the satisfaction for each group. %generalizes the ``covering with outliers'' scenario described above. 
%This notion of fairness was introduced by \cite{BeraG0R14}, where the set of elements is divided into $r$ different colors. One can think of each color class as a subset of voters that identifies as a certain group. We are interested in finding a solution that is allowed to leave certain number of elements uncovered within a group/color class. Thus, the fairness aspect is captured by the constraint that the solution is required to cover \emph{at least} a certain number of elements of each color.

Now we consider an orthogonal generalization of \maxcov. Note that an optimal solution for \maxcov may leave many elements uncovered. However, such a solution may be deemed \emph{unfair} if the elements are divided into multiple colors (representing, say, people of different demographic groups), and the set uncovered elements are biased against a specific color. To address these constraints, the following generalization of \maxcov, which we call \fmaxcov (F stands for ``fair''), has been studied in the literature. Here, we are given a set system $(U, \cF)$, a coloring function $\chi : U \rightarrow [r]$, a coverage requirement function $t: [r] \rightarrow \mathbb{N}$, and an integer $k$; and the goal is to find a subset $\cF' \subseteq \cF$ of size at most $k$ such that, for each $i \in [r]$, the union of elements in $\cF'$ is at least $t(i)$ (or $t_i$). 

Since \fmaxcov is a generalisation of \maxcov, it inherits all the lower bounds known for \maxcov. Furthermore, we can mimic the algorithm for \maxcov (\psc) parameterized by $t$ (where you want to cover at least $t$ elements with $k$ sets)~\cite{DBLP:journals/ipl/Blaser03} to obtain an algorithm for \partsc parameterized by $\sum_{j\in [r]}t_j$. However, the problem is  {\sf NP}-hard even when $t_j \leq 1$,  $j\in [r]$, via a simple reduction from {\sc Set Cover}.
%$\max_{j\in [r]}t_j$. 
%This is due to a simple reduction from the {\sc Set Cover} problem (each element of universe has a distinct color, thus, $|\cU|$ colors and $t_j=1$, for every $j\in [|\cU|]$).  This implies that \partsc does not admit an algorithm with running time $(n+m)^{f(s)}$, where $s=\max_{j\in [r]}t_j$, for any function $f$ depending only on $s$. 

\fmaxcov has been studied under multiple names in the approximation algorithms literature; however much of the focus has been on approximating the \emph{size} of the solution, rather than the coverage. Notable exception include Chekuri et al.\ \cite{DBLP:journals/jco/ChekuriIQVZ22} who gave a ``bicriteria'' approximation, that outputs a solution of size at most $\Oh(\nicefrac{\log r}{\epsilon})$ times the optimal size, and covers at least $(1-\nicefrac{1}{e}-\epsilon)$ fraction of the required coverage of each color. Very recently, Bandyapadhyay et al. \cite{BandyapadhyayFM23} recently designed an \FPTAS for \fmaxcov for the set systems of frequency $2$, running in time $2^{\Oh(\frac{r k^2\log k}{\epsilon})} \cdot (m+n)^{\Oh(1)}$. % as discussed above. 
We obtain the following result on \fmaxcov.

\begin{theorem} \label{thm:intro-fmaxcov}
	%\begin{enumerate}
	%\item 
	There exists a randomized \FPTAS for \fmaxcov running in time $\lr{dr \lr{\frac{\log k}{\epsilon}}^r}^{\Oh(k)} \cdot (m+n)^{\Oh(1)}$, on set systems with frequency bounded by $d$. 
	%\item There exists an \FPTAS for \fmaxcov parameterized by $k, d$, and $\epsilon$, on $K_{d,d}$-free set systems.
	%\end{enumerate}
\end{theorem}

Note that this result generalizes the result of \cite{BandyapadhyayFM23} to frequency-$d$ set systems, and in the case of $d=2$, our running time is faster than that of \cite{BandyapadhyayFM23} (albeit our algorithm is randomized). 
%and even $K_{d,d}$-free set systems. Also, for the case of $d=2$ in the first result, our running time is faster than that in \cite{BandyapadhyayFM23}.

One can also define \emph{fair} version of \satcc in an analogous way, which we call \fmaxsat. In this problem, we are given a CNF-formula $\Phi$, a coloring function $\chi : \cla{\Phi} \rightarrow [r]$, a coverage demand function $t: [r] \rightarrow \mathbb{N}$, and an integer $k$. % the clauses of the CNF-SAT formula are partitioned into $r$ colors, and each color $j$ has an associated demand $t_j$. 
The goal is to find a weight-$k$ assignment that satisfies at least $t(j)$ (also denoted as $t_j$) clauses of each color $j \in [r]$. By combining \Cref{thm:intro-fmaxcov} with a slightly more general version of the reduction theorem (\Cref{thmintro:reduction}) also yields \FPTAS for \fmaxsat with a similar running time.

\subsection{Our New Problem: Combining Matroid and Fairness Constraints} \label{subsec:newprob}

As discussed in the previous subsections, \maxcov has been generalized in two orthogonal directions, namely, matroid constraints on the sets chosen in the solution, and fairness constraints on the elements covered by the solution. Although the corresponding variants of \satcc have not been studied in the literature, we mentioned that our techniques readily imply \FPTAS{}es for these problems for many ``sparse'' formulas. Given this, the following natural question arises.

\begin{tcolorbox}[colback=white!5!white,colframe=gray!75!black]
	Can we find good approximations for the variants of \satcc (resp. \maxcov) that combines the two orthogonal generalizations, namely, matroid constraint on the variables assigned $1$, and fairness constraints on the satisfied clauses (resp.\ matroid constraint on the sets chosen in the solution, and fairness constraints on the elements covered)?
	%What is the tractability boundary of \FPTAS{}es for \satcc/\maxcov and its variants? Can we obtain 
	%Does there exist an {\sf FPT} approximation algorithm for \partsc that outputs a subfamily ${\cF'}\subseteq \cF$ of size at most $k$ such that the number of elements of each color $1 \le j \le r$ covered by $\cF'$, is at least $(1-\epsilon) t_j$, when the frequency of each element is bounded by $d$? 
\end{tcolorbox}

In the following, we formally define the common generalization of \mmaxsat and \fmaxsat, which we call \mfmaxsat. 

\begin{tcolorbox}[colback=white!5!white,colframe=gray!75!black]
	\mfmaxsat
	\\\textbf{Input.} A CNF-SAT formula $\Phi$ where the clauses $\cla{\Phi}$ of $\Phi$ are partitioned into $r$ colors. Each color $j \in [r]$ has an associated demand $t_j$. Additionally, we are provided the independence oracle to a matroid $\cM = (\var{\Phi}, I)$ of rank $k$.  
	\\\textbf{Question.} Does there exist an assignment $\Psi: \var{\Phi} \to \LR{0, 1}$, such that
	\begin{itemize}
		\setlength{\itemsep}{-2pt}
		\item The number of clauses satisfied by $\Psi$ of color $j$ is at least $t_j$, for each $j \in [r]$,
		\item The set of variables assigned $1$ must be independent in $\cM$, i.e., $\Psi^{-1}(1) \in I$.
	\end{itemize}
	%Does there exist an {\sf FPT} approximation algorithm for \satcc on $K_{d,d}$-free formulas  with fairness (as defined in Section~\ref{subsec:fairness}) and matroid constraints?
	%Does there exist an {\sf FPT} approximation algorithm for \partsc that outputs a subfamily ${\cF'}\subseteq \cF$ of size at most $k$ such that the number of elements of each color $1 \le j \le r$ covered by $\cF'$, is at least $(1-\epsilon) t_j$, when the frequency of each element is bounded by $d$? 
\end{tcolorbox}
In the special case where the CNF-SAT formula is monotone (i.e., does not contain negated literals), we obtain \mfmaxcov, which generalizes all the variants of \maxcov discussed earlier. We obtain the following result for \mfmaxcov.

\begin{theorem} \label{thm:intro-mfmaxcov}
	%\begin{enumerate}
	%\item 
	There exists a randomized \FPTAS for \mfmaxcov on set systems with maximum frequency $d$, that runs in time $\lr{\frac{d \log k}{\epsilon}}^{\Oh(kr)} \cdot (m+n)^{\Oh(1)}$ and returns a $(1-\epsilon)$-approximation with at least a constant probability.
	%\item There exists a deterministic \FPTAS for \mfmaxcov on $K_{d, d}$-free set systems, that runs in time $f(k, d, \epsilon) \cdot (m+n)^{\Oh(1)}$, and returns a $(1-\epsilon)$-approximation.
	%\end{enumerate}
\end{theorem}

Finally, by reducing \mfmaxsat on $d$-CNF formulas to \mfmaxcov with frequency $d$ set systems, using the randomized reduction, and then using the results of \Cref{thm:intro-mfmaxcov}, we obtain our most general result, as follows.
%Our most general result is the following.

\begin{theorem} \label{thm:intro-general}
	%\begin{enumerate}
	%\item 
	There exists a randomized \FPTAS for \mfmaxsat on $d$-CNF formulas, that runs in time $\lr{\frac{d \log k}{\epsilon}}^{\Oh(kr)} \cdot (m+n)^{\Oh(1)}$ and returns a $(1-\epsilon)$-approximation with at least a constant probability.
	%\item There exists a radnomized \FPTAS for \mfmaxcov on $K_{d, d}$-free CNF formulas, that runs in time $f(k, d, \epsilon) \cdot (m+n)^{\Oh(1)}$, and returns a $(1-\epsilon)$-approximation with at least a constant probability.
	%\end{enumerate}   
\end{theorem}

We give a summary of how the various problems are related to each other, and a comparison of our results with the literature in \Cref{fig:hierarchy}.

\begin{figure}
	\centering
	\includegraphics[scale=0.65]{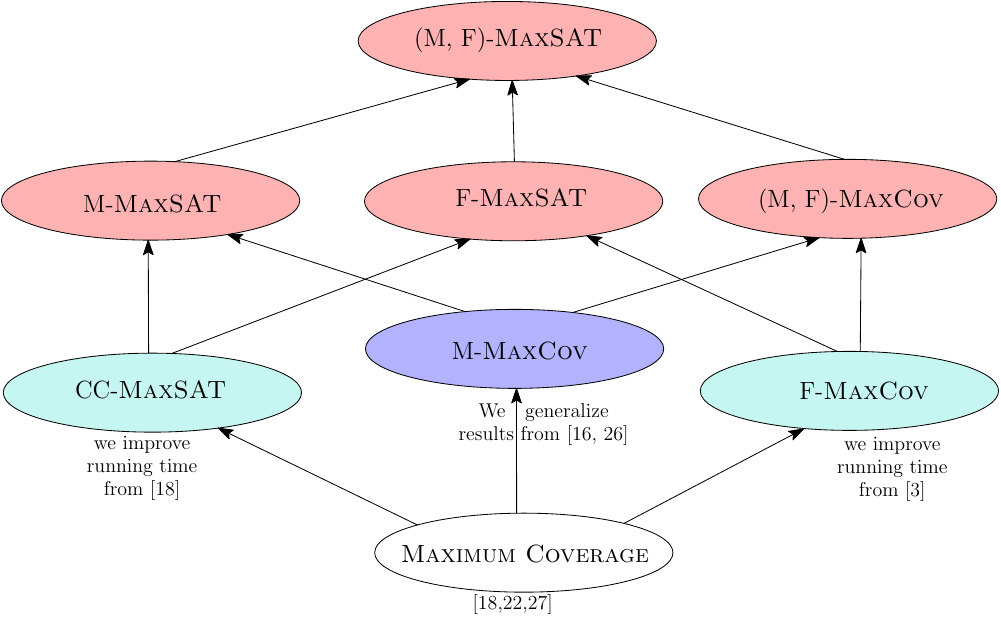}
	\caption{If there is an arrow of the form $A\rightarrow B$, then problem $B$ generalizes problem $A$. \FPTAS{}es for the problems in {\color{red}red} bubbles are not known in the literature, and we study in this paper. For all the other problems \FPTAS{}es are known in the literature for some cases. This paper improves the results in {\color{cyan}cyan} and {\color{blue}blue}.} \label{fig:hierarchy}
\end{figure}

\subsection{Related Results}
\vccc or {\sc Partial Vertex Cover} has been extensively studied in Parameterized Complexity. In this problem we are given a graph and the task is to select a subset of $k$ vertices covering as many of the edges as possible. The problem is known to be approximable within $0.929$ and is hard to approximate  within $0.944$, assuming UGC~\cite{DBLP:journals/corr/abs-1810-03792}. 
\vccc is known to be W[1]-hard~\cite{GuoNW07}, parameterized by $k$, but admits \FPT algorithms on planar graphs, graphs of bounded degeneracy,  \kddfree graphs, and bipartite graphs, parameterized by $k$~\cite{AminiFS11,FominLRS11,KoanaKNS22}. Indeed, it is among the first problems to admit \FPTAS~\cite{DBLP:journals/cj/Marx08,DBLP:journals/corr/abs-1810-03792,DBLP:journals/jair/SkowronF17}. It is also known to have ``lossy kernels''~\cite{DBLP:journals/cj/Marx08,LokshtanovPRS17}, a lossy version of classical kernelization.

Bera et al. \cite{DBLP:journals/tcs/BeraG0R14} considered the special case of \textsc{Partition Vertex Cover}, where the set of edges of a graph are divided into $r$ colors, and we want to find a subset of vertices that covers at least a certain number of edges from each color class. For this problem, they gave a polynomial-time $\Oh(\log r)$-approximation algorithm. Hung and Kao \cite{DBLP:journals/algorithmica/HungK22} generalized this to \fmaxcov, and gave a $\Oh(d \log r)$-approximation, where each element of the universe is contained in at most $d$ sets (i.e., $d$ is the maximum \emph{frequency}). Bandyapadhyay et al. \cite{DBLP:journals/corr/abs-2308-15842} studied this problem under the name of \textsc{Fair Covering}, and designed a $\Oh(d)$-approximation, but their running time is \textsf{XP} in the number of colors. Chekuri et al. \cite{DBLP:journals/jco/ChekuriIQVZ22} designed a general framework for \fmaxcov, yielding tight approximation guarantees for a variety of set systems satisfying certain property; in particular, they improve the approximation guarantee for frequency-$d$ set systems to $\Oh(d + \log r)$, which is tight in polynomial time. 
%\input{overview}

%%%%%%Remove in long version
%The formal proofs can be found in the appendix (we mention the section numbers of the appendix at appropriate places in the overview and mark them in blue color and $\clubsuit$).

\section{Overview of Our Results and Techniques}
\subsection{Reduction from \satcc to \maxcov: An overview of \Cref{thmintro:reduction}} \label{subsec:overview-satcov}
This theorem is essentially a randomized \emph{approximation-preserving reduction} from \satcc to \maxcov. Given an instance $\cI= (\Phi, k)$ of \satcc, we first compute a random assignment $\Psi$ that assigns a variable independently to be $1$ with probability $p = \epsilon/2$ and $0$ with probability $1-p$. Let $V^*$ be the set of at most $k$ variables set to be $1$ by an optimal assignment $\Psi^*$. It is straightforward to see that, the probability that all the variables in $V^*$ are set to be $1$ by the random assignment $\Psi$ is $p^k$ -- we say that this is the good event $\mathcal{G}$. Now, consider a clause that is satisfied negatively by $\Psi^*$, i.e., a clause $C$ that contains a negative literal $\neg x$ and $\Psi^*(x) = 0$. It is also easy to see that, conditioned on the good event $\mathcal{G}$, the probability that such a clause $C$ is also satisfied negatively by $\Psi$ is at least $1-p$. Thus, the expected number of clauses that are satisfied negatively by $\Psi$, conditioned on $\mathcal{G}$, is at least $1-p$ times the number of clauses satisfied negatively by $\Psi^*$. Markov's inequality implies that, with probability at least $1/2$, the actual number of such clauses is close to its expected value. Thus, conditioned on $\mathcal{G}$, and the previous event, we can focus on the positively satisfied clauses (note that the probability that both of these events occur is at least $1/2 \cdot (\epsilon/2)^k$. To this end, we can eliminate all the negatively satisfied clauses, and we can also prune the remaining clauses by eliminating any negative literals and the variables that are set to $0$ by $\Psi$. Thus, all the remaining clauses only contain positive literals, which can be seen as an instance $\cI'$ of \maxcov. Furthermore, conditioned on $\mathcal{G}$, the variables set to $1$ by $\Phi^*$ correspond to a family $\cF^*$ of size $k$, and the elements covered by $\cF^*$ correspond to the set of clauses satisfied only positively by $\Phi^*$. Thus, if we find a $(1-\epsilon)$-approximate solution to $\cI'$, and set the corresponding variables to $1$, and the rest of the variables to $0$, then we get a weight-$k$ assignment that satisfies at least $(1-\epsilon) \cdot \OPT_{sat}$ clauses. Note that this reduction, combined with the algorithm of \cite{DBLP:conf/soda/0001KPSS0U23} gives the proof of \Cref{introcor:satcc}. 
%an \FPTAS for \satcc on $K_{d,d}$-free formulas that improves upon the running time from Jain et al.\cite{DBLP:conf/soda/0001KPSS0U23} -- to obtain this result, we first use this randomized reduction followed by the FPT-AS from the same paper for \maxcov on $K_{d,d}$-free set systems. 

Furthermore, this reduction is robust enough that it can accommodate the fairness constraints on the clauses, as defined above. To be precise, one can give a similar reduction from $\mathcal{C}$-\textsc{MaxSAT} to $\mathcal{C}$-\textsc{MaxCov}, where $\mathcal{C} \in \LR{\textsc{M, F, (M, F)}}$ -- note that when we have fairness constraints, the success probability now becomes $(r/\epsilon)^{\Oh(k)}$.
%More specifically, we can define a fair version of \satcc, \textsc{Partition \satcc}\todo{why red}, where the clauses of the given CNF-SAT formula are divided into $r$ colors, and the goal is to find a weight-$k$ assignment that satisfies at least $t_j$ clauses of color $j$, for each $1 \le j \le r$.\todo{make the definition of \textsc{Partition \satcc} consistent with Maximum Coverage definition above.} 
%It is relatively straightforward to extend the approach in the previous paragraph, to give a reduction from \red{\textsc{Partition \satcc}} to \textsc{Partition Maximum Coverage}; albeit the success probability now becomes $(r/\epsilon)^{\Oh(k)}$. 
Essentially, these reductions translate a constraint on the variables set to $1$ (for \satcc and variants), to the corresponding family of sets (for \maxcov and variants). 
%For example, consider the problem \mmaxsat. Here, in addition to the \satcc input $(\Phi, k)$,  where matroid $\cM = (V, \cI)$ of rank $k$, where $V$ is the set of variables, and the set of variables assigned be $1$ is required to be an independent set in $\cM$. Then, the same approach in fact gives a reduction from \red{\textsc{Matroid Constrained \satcc}} to \red{\textsc{Matroid Constrained \maxcov}} (or their \textsc{Partition} generalizations). 
Thus, at the expense of a multiplicative $(r/\epsilon)^{\Oh(k)}$ factor in the running time, we can focus on \maxcov and its variants, which is what we do in this section, as well as in the paper. As a warm-up, we start in \Cref{subsec:overview-branching} with the vanilla \maxcov on frequency-$d$ set systems (where our algorithms \emph{do not} improve over the known algorithms in the literature), and give a complete formal proof. Then, we will gradually introduce the ideas required to handle fairness (\Cref{subsec:overview-bucketing}) and matroid (\Cref{subsec:overview-matroid}) constraints -- first separately, and then together. Finally, in \Cref{subsec:overview-generalizations}, we briefly discuss the ideas required to these results to $K_{d,d}$-free set systems and multiple matroid constraints. %\red{Due to the space constraints, in the short version we will only keep the full proof of Theorem~\ref{thmintro:reduction} in Section~\ref{sec:reduction}. Other technical details are deferred to the full version.}

\subsection{Deterministic and Randomized Branching using a Largest Set} \label{subsec:overview-branching}
To introduce our ideas in a clean and gradual way, we start with the simplest setting of \maxcov where the maximum frequency of the elements is bounded by $d$. Recall that we are given an instance $(U, \cF, k)$ and the goal is to find a sub-family of $\cF$ of size $k$ that covers the maximum number of elements. For any sub-family $\cR \subseteq \cF$, let $U(\cR)$ denote the subset of elements covered by $\cR$, and $\OPT_k(\cR)$ denote the maximum number of elements that can be covered by a subset of $\cR$ of size $k$. Further, for a set $S \in \cF$, we denote by $\cF - S$, the family obtained by removing $S$, as well as the elements of $S$ from each of the remaining sets. Our approach is inspired by the approaches of Skowron and Faliszewski~\cite{DBLP:journals/jair/SkowronF17} and Manurangsi~\cite{DBLP:journals/corr/abs-1810-03792} who show that $\Oh(kd/\epsilon)$ sets of the largest size is guaranteed to contain a $(1-\epsilon)$-approximate solution. This naturally begs the question, ``\emph{why not start by adding the largest set into the solution?}'' (in a sense, the following presentation is closer in spirit to Jain et al.~\cite{DBLP:conf/soda/0001KPSS0U23}.) Let us inspect this question more closely. Let $L$ be a largest set in $\cF$. By looking at the contribution of coverage of each set in an optimal solution, say $\Oh$, we can easily see that $|L| \ge \frac{\OPT_k(\cF)}{k}$. We say that a set $S \in \cF$ is \emph{heavy} w.r.t. $L$ if $|L \cap S| \ge \frac{\epsilon |L|}{k}$ (note that $L$ is heavy w.r.t. itself). However, since the frequency of each element is bounded by $d$, each element in $L$ can appear in at most $d$ (in fact, $d-1$) sets $L \cap S$ for different $R \in \cF$. This implies that at most $\frac{kd}{\epsilon}$ sets in $\cF$ are heavy w.r.t. $L$.

\smallskip 
{\bf Algorithm.} Our algorithm simply branches on the sets in $\cH(L)$, which is the family of heavy sets w.r.t. $L$. Specifically, in the branch corresponding to a heavy set $S \in \cH(L)$, we include it in the solution, and recursively call the algorithm on the residual instance $(U \setminus S, \cF - S, k-1)$. If any of the sets in $\Oh$ is heavy w.r.t. $L$, then in the branch corresponding to such a set yields an approximate solution via induction. The main idea is that, if no set in $\Oh$ is heavy w.r.t. $L$, then the branch corresponding to $L$ yields a good solution. This is justified as the effect of any of the $k-1$ sets in $\Oh$ is too small. For the sake of clarity, we formally analyze this algorithm below via induction.

We want to show that, for a given input $(U', \cF', k')$ the recursive algorithm returns a family $\cR \subseteq \cF'$ of size $k'$ such that $|U'(\cR)| \ge (1-\epsilon) \cdot \OPT_{k'}(\cF')$. The base case for $k' = 0$ is trivial, since the algorithm returns an empty set. Suppose that the claim is true for all inputs with budget $k-1$, and we want to prove it for $(U, \cF, k)$. Let $\Oh$ denote an optimal solution of size $k$ with $|\OPT_k(\cF)| = |U(\Oh)|$. 

\smallskip
\noindent
{\bf Approximation Ratio in the Easy Case.} If $\Oh \cap \mathcal{H}(L) \neq \emptyset$, then there exists a branch corresponding to a set $S \in \Oh \cap \mathcal{H}(L)$. This is the \emph{easy case} (for the analysis). In this case, $\OPT_{k-1}({\cF} - S) = \OPT_{k}({\cF})-|S|$, and hence the approximation ratio is: 

\begin{eqnarray*}
	\frac{|S| + (1-\epsilon) \OPT_{k-1}({\cF} \setminus S) }{\OPT_{k}({\cF})}  &= &  \frac{|S| + (1-\epsilon) (\OPT_{k}({\cF})-|S|) }{\OPT_{k}({\cF})} \\
	&\geq &  \frac{(1-\epsilon) (\OPT_{k}({\cF})}{\OPT_{k}({\cF})} = (1-\epsilon) 
\end{eqnarray*}

\smallskip
\noindent
{\bf Approximation Ratio in the Hard Case.} The more hard case (for analysis) is when $\Oh \cap \mathcal{H}(S) = \emptyset$. In this case, we argue that the branch that includes the element $L$ is good enough. As in the easy case, we first lower bound the value of $\OPT_{k-1}({\cF} - L)$. By counting the unique contributions to the solution, there exists a \emph{light} set $S_l \in \Oh$ such that for $\Oh' = \Oh \setminus \LR{S_l}$, it holds $|U'(\Oh')| \ge \frac{k-1}{k} \cdot \OPT_k(\cF)$. Because no set in $\Oh$ is heavy w.r.t. $L$, it follows that for each $R \in \Oh'$, it holds that $|R \cap L| < \frac{\epsilon |L|}{k}$, and therefore by counting it holds that $|U(\Oh') \cap L| < \epsilon \cdot |L|$. Therefore, 
\begin{eqnarray*}
	\OPT_{k-1}({\cF} \setminus L)   \geq |U'(\Oh') \setminus L| \geq |U(\Oh')| - |U(\Oh') \cap L| \geq  \frac{k-1}{k} \cdot \OPT_k(\cF) - \epsilon \cdot |L|.
\end{eqnarray*}
Therefore, the approximation ratio of the branch that includes $L$ is as follows.

\begin{eqnarray*}
	\frac{|L| + (1-\epsilon) \OPT_{k-1}({\cF} - L) }{\OPT_{k}({\cF})}  
	&= &  \frac{|L| + (1-\epsilon) \left(\frac{k-1}{k} \cdot \OPT_k(\cF) - \epsilon \cdot |L|\right) }{\OPT_{k}({\cF})} \\
	& \geq & \frac{|L| + (1-\epsilon) \left(\frac{k-1}{k} \cdot \OPT_k(\cF)\right) - \epsilon  |L| }{\OPT_{k}({\cF})} \\
	& = &  \frac{(1-\epsilon) |L| + (1-\epsilon) \left(\frac{k-1}{k} \cdot \OPT_k(\cF)\right)  }{\OPT_{k}({\cF})} \\
	&\geq &  \frac{(1-\epsilon) (\OPT_{k}({\cF})}{\OPT_{k}({\cF})} = (1-\epsilon) 
\end{eqnarray*}
The second last inequality holds from the fact that $|L| \ge \frac{\OPT_k(\cF)}{k}$. 
% Combining these bounds, we can show that in the residual instance $(U \setminus M, \cF', k-1)$, the family $\Oh'$ covers ``sufficiently many'' elements \todo{explain how many}. Since $\Oh'$ is a candidate solution of size $k-1$, $\OPT_{k-1}(\cF')$ is at least \red{so much}. Then, by induction, it follows that the recursive call corresponding to this branch will return a solution $\cR'$, that, when combined with $M$, covers at least $(1-\epsilon) \cdot \OPT_k(\cF)$ elements. 

This leads to a deterministic $(1-\epsilon)$-approximation algorithm with running time $((\frac{kd}{\epsilon})^k) \cdot (n+m)^{\Oh(1)}$.

\paragraph{Insight into the probabilistic branching.}
A closer inspection of the analysis reveals that the reason $L$ may not a good choice is that the sets of $\Oh$ \emph{together} cover more than a certain threshold fraction of elements covered by $L$. We utilize this idea through a smoothening process that captures the effect of the size of the intersection of a set $S$ with $L$ in a more nuanced manner. Let us define the weight $h_L(S)$ of a set $S \in \cF \setminus \LR{L}$ as $h_L(S) = \frac{|S \cap L|}{|L|}$. Our algorithm now instead does ``randomized branching'', i.e., it samples one set to be included in the solution according to some probability distribution, and then continues recursively. Note that the single run of the algorithm finishes in polynomial time. The probability distribution used by the algorithm is as follows: the set $L$ is sampled with probability $1/2$, and any other set $S \in \cF \setminus \LR{L}$ is sampled with probability proportional to its weight $h(S)$ (the constant of proportionality is chosen such that this is a valid probability distribution, that is, the probabilities sum up to $1$). In particular, observe that $\sum_{S \in \cF \setminus \LR{L}}h_L(S) \leq  \frac{d |L|}{|L|}=d$. Thus, the probability of selecting $S$ is at least $\frac{h_L(S)}{2d}$. Note that due to the way the weights $h_L(S)$ are defined, the sets with a large intersection with $L$ have a greater chance of being sampled, as compared to the sets with a small intersection with $L$.

We will show that the algorithm returns a $(1 -\epsilon)$-approximate solution  with probability at least $(\frac{\epsilon}{2d})^k$, and runs in polynomial time. This implies that by repeating the algorithm $(\frac{2d}{\epsilon})^k$ times, we obtain a $(1-\epsilon)$-approximation with probability at least a positive constant. This leads to a randomized algorithm with running time $\Oh^*((\frac{2d}{\epsilon})^k)$.

The proof is again by induction. We want to show that, for any input $(U', \cF', k')$, the algorithm returns a solution $\cR \subseteq \cF'$ of size $k'$ such that $|U'(\cR)| \ge (1-\epsilon) \cdot \OPT_{k'}(\cF')$, with probability at least $(\frac{\epsilon}{2d})^{k'}$. We reuse much of the notation from the previous analysis. Let $\Oh$ be an optimal solution of size $k$. First, the case when $L \in \Oh$, since our algorithm samples and includes $L$ in the solution with probability $1/2$. Then, conditioned on this event (that is, $L$ being sampled), the approximation ratio analysis proceeds similarly to the \emph{easy case} of the previous analysis. By induction, the recursive algorithm returns a $(1-\epsilon)$-approximate solution with probability at least $(\frac{\epsilon}{2d})^{k-1}$. Thus, overall, the algorithm returns a $(1-\epsilon)$-approximation with probability at least $\frac{1}{2}(\frac{\epsilon}{2d})^{k-1} \geq (\frac{\epsilon}{2d})^k$.

Now suppose $L \not\in \Oh$. As before, let $S_l \in \Oh$ be a light set as defined earlier, and $\Oh' = \Oh \setminus \LR{S_l}$.  We analyze by considering the following two cases: either (i) $|U(\Oh') \cap L| \le \epsilon \cdot |L|$, or (ii) $|U(\Oh') \cap L| > \epsilon \cdot |L|$. 

In case (i), we are effectively in the same situation as the \emph{hard case} of the previous analysis -- as before, the algorithm samples $L$ with probability at least $1/2$, and as argued in the \emph{hard case}, conditioned on the previous event, the branch corresponding to $L$ returns a $(1-\epsilon)$-approximate solution, but now with probability at least $(\frac{\epsilon}{2d})^{k-1}$ by induction. Therefore, we obtain an $(1-\epsilon)$-approximate solution with probability at least $\frac{1}{2}(\frac{\epsilon}{2d})^{k-1} \geq (\frac{\epsilon}{2d})^k$. 

In case (ii), we have that $|U(\Oh') \cap L| > \epsilon \cdot |L|$. Notice that, 
\begin{eqnarray*}
	\sum_{S\in \Oh'} |S\cap L|  > \epsilon. 
\end{eqnarray*}
This implies that 
\begin{eqnarray*}
	\sum_{S\in \Oh'} h_L(S)= \sum_{S\in \Oh'} \frac{|S\cap L|}{|L|} \geq |U(\Oh') \cap L| > \epsilon |L|. 
\end{eqnarray*}

Therefore, the total weight of the sets in $\Oh'$ is at least $\epsilon$. Therefore, the probability that the algorithm samples a set from $\Oh'$ is at least $\frac{\epsilon}{2d}$. Conditioned on this event, the approximation ratio analysis now proceeds as in the \emph{easy case}, and the algorithm returns a $(1-\epsilon)$-approximate solution with probability at least $\frac{\epsilon}{2d}(\frac{\epsilon}{2d})^{k-1} = (\frac{\epsilon}{2d})^k$.

\medspace

%\noindent\textbf{Description of Bucketing Trick.} 
%Our first idea is to use  a knapsack-style bucketing technique to group the vertices of $A$ into \emph{approximate equivalence classes}, called \emph{bags} for short. At a high level, all the vertices belonging to a particular bag contain \emph{approximately equal} (i.e., within a factor of $(1+\epsilon)$) number of neighbors of each color in $B$. Thus, in isolation, any two vertices $v_1$ and $v_2$ that belong to the same bag are interchangeable, since we tolerate an $\epsilon$-factor loss in the coverage. However, consider a solution $S' \subseteq A$ of size at least $2$. Now, consider a vertex $v_1 \in S'$, and another vertex $v_2$ that belongs to the same bag as $v_1$. In this case, $v_1$ and $v_2$ may not be interchangeable w.r.t.\ the solution $S'$, since $v_2$ may cover many common elements as the rest of the sets in the solution $S' \setminus \LR{v_1}$. Nevertheless, we are able to carefully design approximation algorithms that can leverage this approximate interchangeability in ``good situations''.

\subsection{Handling fairness constraints via the Bucketing trick} \label{subsec:overview-bucketing}

The aforementioned idea of prioritizing the largest set $L$, or sets that are heavy w.r.t. it, fails to generalize when we have multiple coverage constraints in \fmaxcov. This is simply because there is no notion of ``the largest set'' even when we want to cover elements of \emph{two} different colors, each with different coverage requirements. To handle such multiple coverage constraints, our idea is to use multidimensional-knapsack-style bucketing technique to group the sets of $\cF$ into \emph{approximate equivalence classes}, called \emph{bags} for short. At a high level, all the vertices belonging to a particular bag contain \emph{approximately equal} (i.e., within a factor of $(1+\epsilon)$) number of elements of \emph{all} of the $r$ colors. Thus, in isolation, any two sets $L_1$ and $L_2$ belonging to the same bag are interchangeable, since we tolerate an $\epsilon$-factor loss in the coverage. Since the total number of bags can be shown to be $\lr{\frac{\log k}{\epsilon}}^{\Oh(rk)}$, and hence we can ``guess'' a bag ${\cal B}$ containing a set from a solution $\Oh$. However, due to different amount of overlap with an optimal solution $\Oh$, two sets $L_1, L_2 \in {\cal B}$ may not be interchangeable w.r.t. $\Oh$. That is, if $L_1 \in \Oh$, then $\Oh \setminus \LR{L_1} \cup \LR{L_2}$ may not be a good solution. However, assuming that we have correctly guessed a bag that intersects with $\Oh$, we can then select a set $L \in {\cal B}$, and define the heavy sets (for deterministic algorithm) or weights $h_L(\cdot)$ (for the randomized algorithm) w.r.t. $L$. Note that, since we have multiple coverage constraints, we cannot simply look at the total size of the intersection $|S \cap L|$. Instead, we need to tweak the notion of \emph{heaviness} that takes into account the number of elements of each color in the intersection $S \cap L$. To summarize, we need two additional ideas to handle multiple colors in \fmaxcov: (1) ``guessing'' over buckets, and (2) a suitable generalization of the notion of heaviness. Modulo this, the rest of the analysis is again similar to the \emph{easy} and \emph{hard} cases as before. Using these ideas, we can prove the first part of \Cref{thm:intro-fmaxcov}.

\subsection{Handling Matroid Constraints.} \label{subsec:overview-matroid}

First, we consider \mmaxcov (note that this is an orthogonal generalization of \maxcov, without multiple coverage constraints), where the solution is required to be an independent set in the given matroid $\cM= (\cF, I)$ of rank $k$. We assume that we are given an \emph{oracle access} to $\cM$ in the form of an algorithm that answers the queries of the form ``Is $\cR$ an independent set?'' for any subset $\cR \subseteq \cF$. Let us revisit the initial deterministic \FPTAS for \maxcov and try to generalize it to \mmaxcov. Recall that this algorithm branches on each set $S \in \cH(L)$, where $L$ is a largest set. The analysis of easy case goes through even in presence of the matroid constraint, since we branch on a set from an optimal solution $\Oh$. However, in the hard case, our analysis replaces a $S_l \in \Oh$ with $L$, and argues that the branch corresponding to $L$ returns a $(1-\epsilon)$-approximate solution. However, this does not work for \mmaxcov, since $(\Oh \setminus \LR{S_l}) \cup \LR{L}$ may not be an independent set in $\cM$. To summarize, although $L$ handles the coverage constraints (approximately), it may fail to handle the matroid constraint. In fact, it may just so happen that $L$ is not a good set \emph{at all}, in the sense that, for \emph{any} set $S \in \Oh$, $(\Oh \setminus \LR{S}) \cup \LR{L}$ is not independent in $\cM$, which is crucial for the induction to go through.
%Here, we still need to use the bucketing idea due to multiple coverage requirements. However, due to matroid constraints, two vertices in the same bucket may not be interchangeable, \emph{even if} we ignore the coverage overlap with the rest of the solution as described in the paragraph above. To see this issue carefully, consider two vertices $x$ and $y$ from the same bag $A'$, and suppose $x \cup S'$ is a solution for the given instance. In this situation, even if a vertex $y$ from the same bag as $x$, has relatively low neighborhood intersection with $S'$, $y \cup S'$ covers approximately the same number of elements as $x \cup S'$. However, $y \cup S'$ may not be an independent set in the given matroid $\cM = (A, I)$, and thus, is not a solution. 

To solve this issue, we resort to the bucketing idea as in the fair coverage case (thus, the subsequent arguments generalize to $\mfmaxcov$ in a straightforward manner; although let us stick to the special case of \mmaxcov for now). Indeed, branching w.r.t. the largest set $L$ is an overkill -- it suffices to pin down a bag $\mathcal{B}$ containing a set in $\Oh$ (it does not even have to be the largest set), by ``guessing'' from $\Oh\lr{\frac{\log k}{\epsilon}}$ bags. However, we again cannot select an arbitrary set $S \in \cB$ and define heavy sets w.r.t. $S$, precisely due to the matroid compatibility issues mentioned earlier. Therefore, we resort to the idea of \emph{representative sets} from matroid theory \cite{DBLP:conf/acid/Marx06} \footnote{Although this is a powerful hammer in its full generality---which we do use to handle multiple matroid constraints---our specialized setting lets us use much simpler arguments to handle single matroid constraint in \mmaxcov/\mfmaxcov.} Assuming our guess for $\cB$ is correct, there exists some $S \in \cB \cap \Oh$. However, we cannot further ``guess'' $S$, since the size of the bag may be too large. Instead, we compute a inclusion-wise maximal independent set $\cB' \subseteq \cB$. Note that the size of $\cB'$ is at most $k$, and it can be computed using polynomially many queries to the independence oracle. However, it may very well happen that $S \not\in \cB'$. Nevertheless, using matroid properties, we can argue that, there exists a set $S' \in \cB'$, such that $(\Oh \setminus \LR{S}) \cup \LR{S'}$ is an independent set. Thus, $\cB'$ is a \emph{representative set} of $\cB$. Furthermore, since both $S$ and $S'$ come from the same bag, they cover approximately the same number of elements. Thus, our modified deterministic algorithm works as follows. First, it computes maximal independent set $\cB' \subseteq \cB$, and for each $S' \in \cB'$, it computes the heavy family $\cH(S')$. Then, it branches over all sets in $\bigcup_{S' \in \cB'} \cH(S')$. If one of the branches corresponds to branching on a set from $\Oh$, then the analysis is similar to the easy case. Otherwise, we know that $(\Oh \setminus \LR{S}) \cup \LR{S'}$ is an $(1-\epsilon)$-approximate solution that is also an independent set. Therefore, the branch corresponding to $S'$ yields the required $(1-\epsilon)$-approximation. We can improve the running time via doing a randomized branching in two steps: first we pick a set $S'' \in \cB'$ uniformly at random, and then we perform the probabilistic branching using the weights $h_{S''}(\cdot)$. 

Both deterministic and randomized variants incur a further multiplicative overhead of $(\frac{k \log k}{\epsilon})^k$ due to first guessing a bag, and then computing a representative set $\cB' \subseteq \cB$ of the bag, and thus do not improve over the results of Sellier~\cite{DBLP:conf/esa/Sellier23} in terms of running time for \mmaxcov. However, this idea naturally generalizes to \mfmaxcov, with the appropriate modifications in bucketing (as mentioned in the previous paragraph) to handle the multiple coverage requirements of different colors. This leads to the proof of \Cref{thm:intro-mfmaxcov}. 

%The formal proofs can be found in \bluesuit{Section 7} in the attached full version.
%%%%%%Remove in long version
%The formal proofs can be found in \bluesuit{Section~\ref{sec:matroidc}} in the appendix.

%To solve this issue, we use the idea of \emph{representative sets} from matroid theory. However, we need to 

%compute a ``representative set'' of the bag $A$.
%Let $R \subseteq A'$ be an \emph{inclusion-wise maximal} independent set. Note that such a set $R$ can be computed in polynomial time using a simple greedy algorithm using via oracle access to $\cM$. Note that $|R| \le k$ since $R \in \cI$. Further, we can use the matroid properties to prove the following property of the set $R$ -- for any $u \in A'$, if there exists some $P \subseteq A$ such that $P \uplus \LR{u} \in \cI$, then there exists some $u' \in R$, such that $P \uplus \LR{u'} \in I$. This property can then be used to replace $u$ with $u'$ in a potential solution $P \uplus \LR{u}$ (of course, we also need to analyze the effect of such a replacement on the coverage). Since we do not know the ``correct'' replacement $u' \in R$, we branch on each element in $R$, which slightly worsens the running time. Thus, we obtain the following generalizations of \Cref{thm:d-hs} and \Cref{thm:introkdd} respectively. 

\subsection{Further Extensions} \label{subsec:overview-generalizations}

The ideas mentioned in the previous subsections can be extended to even more general settings in a couple of ways. First, we describe how to extend the ideas from frequency-$d$ set systems for \maxcov to $K_{d,d}$-free set systems, i.e., set system $(U, \cF)$, where no $d$ sets in $\cF$ contain $d$ elements of $U$ in common. Note that frequency-$d$ set systems are $K_{d+1,d+1}$-free. Then, we describe how the linear algebraic toolkit of \emph{representative sets} can be used to handle multiple (linear) matroid constraints.

\paragraph{$K_{d,d}$-free Set Systems.} Next, we consider $K_{d,d}$-free set systems $(U, \cF)$, where no $d$ sets of $\cF$ contain $d$ common elements of $U$. To design the \FPTAS on $K_{d,d}$-free set systems, we combine the bucketing idea along with the combinatorial properties of $K_{d,d}$-free graphs to bound the number of heavy neighbors of a set. To this end, however, we need to modify the precise definition of \emph{heaviness}  (as in Jain et al.~\cite{DBLP:conf/soda/0001KPSS0U23}). This leads to a somewhat cumbersome branching algorithm that handles colors differently based on their coverage requirement. For colors with small coverage requirement, we highlight the covered vertices using the standard technique of \emph{label coding}\footnote{The technique is better known as \emph{color coding}. However, this creates an unfortunate clash of terminology -- these \emph{colors} have nothing to do with the original colors corresponding to coverage constraints. Bandyapadhyay et al.~\cite{BandyapadhyayFM23} instead use the term ``label coding'', and we also adopt the same terminology}. Now, vertices in a bag cover the elements with the same label and for colors with high coverage requirements, the sizes of the sets in the same bag are ``almost'' the same. Then, we pick an \emph{arbitrary} set $L$ from the bag ${\cal B}$, and we branch on (suitably defined) \emph{heavy} sets w.r.t. $S$. Since the number of heavy sets is bounded by a function of $k, d$, and $\epsilon$, this leads to a deterministic version of \Cref{thm:intro-fmaxcov} to $K_{d,d}$-free set systems. Note that since frequency-$d$ set systems is a special case of this, this implies a deterministic \FPTAS in this case; however with a much worse running time compared to \Cref{thm:intro-fmaxcov}.
This leads to the proof of \Cref{thm:intro-mfmaxcov}.

\paragraph{Handling Multiple Matroid Constraints.}
Our results on \mmaxcov and \mfmaxcov can be generalized to handle multiple matroid constraints on the solution, in the case when the matroids are \emph{linear} or representable\footnote{A matroid $\cM = (E, \cI)$ is representable over a field $\mathbb{F}$ if there exists a matrix $M$ such that there exists a bijection between $E$ and the columns on $M$ with the property that, a subset $E' \subseteq E$ is independent in $\cM$ iff the corresponding set of columns are linearly independent over $\mathbb{F}$.}. In this more general problem, we are given $q$ linear matroids $\cM_1, \cM_2, \ldots, \cM_q$, where $\cM_i = (\cF, I_i)$, each of rank at most $k$, and the solution $S$ is required to be independent in all $q$ matroids, i.e., $S \in \bigcap_{i \in [q]} I_i$. In this case, we can use linear algebraic tools (\cite{DBLP:conf/acid/Marx06,fomin2014efficient}) to compute a representative set of size $qk$ instead of $k$, and the computation requires $2^{\Oh(qk)} \cdot n^{\Oh(1)}$ time. Thus, \FPTAS{}es for this problem now have a factor of $q$ in the exponent. 

Note that our \FPTAS improves upon the polynomial-time approximation guarantee of $1-1/e$ of Calinescu et al.~\cite{CalinescuCPV11} for monotone submodular maximization subject to a matroid constraint, in the special case of $K_{d,d}$-free coverage functions. To the best of our knowledge, this is the largest class of monotone submodular functions and matroid constraints for which the lower bound of $1-1/e$ can be overcome, even in \FPT time. Further, the analogous results to \mfmaxsat generalize these results to maximization of non-monotone/non-submodular functions.

\section{Preliminaries}

For a positive integer $q$, let $[q] \coloneqq \LR{1, 2, \ldots, q}$. 

\paragraph{Equivalent reforumation in terms of dominating set in the incidence graph.}

For the ease of exposition, we recast \fmaxcov to the following problem, and work with this formulation in the rest of the paper.

\begin{tcolorbox}[colback=white!5!white,colframe=gray!75!black]
	\pbdsfull (\pbds)
	\\\textbf{Input:} An instance $\cI = (G, r, f, t, k)$, where
	\begin{itemize}
		\item $G = (A \uplus B, E)$ is a bipartite graph with bipartition $(A, B)$,
		\item $f: B \to [r]$ is a surjective function, and for each $j \in [r]$, we say that $B_j \coloneqq \LR{ \fs \in B: f(\fs) = j }$ is a set of vertices of \emph{color} $j$,
		\item For each color $j \in [r]$, a \emph{coverage requirement} $t_j \ge 0$, and
		%$t: [r] \to \mathbb{N}$, where $t_j \coloneqq t(j)$ is called the \emph{coverage requirement} of color $j \in [r]$, and
		\item a non-negative integer $k$.
	\end{itemize}
	\textbf{Question:} Does there exist a subset $S \subseteq A$, such that (1) $|S| \le k$, and (2) for each $j \in [r]$, $|N_j(S)| \ge t(j)$?
	Here, $N_j(S) \subseteq B_j$ is the set of vertices of color $j \in [r]$ that are adjacent to at least one vertex in $S$. 
\end{tcolorbox}

Note that \pbds is the ``multiple coverage'' version of the well-studied problem {\sc Red Blue Dominating Set}. We avoid ``Red Blue'' here just to avoid confusion with our color classes. Observe that finding a subfamily of size $k$ that covers at least $t_j$ elements of the color $j$ is equivalent to finding a set $S\subseteq A$ of size $k$ such that $N_j(S)$ (neighbors of $S$ that are colored $j$) is at least $t_j$ in the incidence graph.

%In the technical sections, we consider the following equivalent formulation of \partsc, which we call \pbdsfull. Note that the two problems can be seen equivalent, since the graph $G = (A \uplus B, E)$ can be taken to be the set-element incidence graph of the set system $(\cU, \cF)$ -- the vertices in $A$ correspond to the sets in $\cF$, the vertices in $B$ correspond to the elements in $\cU$, and an edge between the vertices corresponding to $S \in A$ and $e \in B$, corresponds to the element $e \in \cU$ being contained in the set $S \in \cF$.
%
%\begin{tcolorbox}[colback=white!5!white,colframe=gray!75!black]
%	\pbdsfull
%	\\\textbf{Input:} An instance $\cI = (G, r, f, t, k)$, where
%	\begin{itemize}
	%		\item $G = (A \uplus B, E)$ is a bipartite graph with bipartition $(A, B)$,
	%		\item $f: B \to [r]$ is a surjective function, and for each $j \in [r]$, we say that $B_j \coloneqq \LR{ \fs \in B: f(\fs) = j }$ is a set of vertices of \emph{color} $j$,
	%		\item $t: [r] \to \mathbb{N}$, where $t(j)$ is called the \emph{coverage requirement} of color $j \in [r]$, and
	%		\item a non-negative integer $k$.
	%	\end{itemize}
%	\textbf{Question:} Does there exist a subset $S \subseteq A$, such that (1) $|S| \le k$, and (2) for each $j \in [r]$, $|N_j(S)| \ge t(j)$?
%	\\Here, $N_j(S) \subseteq B_j$ is the set of vertices of color $j \in [r]$ that are adjacent to at least one vertex in $S$. 
%\end{tcolorbox}

%\todo[inline]{For $S\subseteq B$, define $f(S)$.\\TI: What does $f(S)$ mean for $S \subseteq B$?}

\medskip\noindent\textbf{Convention.} For a graph $H$ with bipartition $A' \uplus B'$, we refer to the $A'$ (resp.\ $B'$) as the \emph{left} (resp.\ \emph{right}) side of the bipartition $A' \uplus B'$. We will consistently use standard letters such as $u, v, w$ to refer to vertices on the left, and fraktur letters such $\text{\large $\mathfrak{s, p, q}$}$ to refer to the vertices on the right. We consistently use index $j$ to refer to a color from the range $[r]$, and may often write ``for a color $j$'' instead of ``for a color $j \in [r]$''. Finally, for a coverage requirement function $t$ (resp.\ variations such as $t', \ttt$), and a color $j$, we shorten $t(j)$ to $t_j$ (resp.\ $t'_j, \ttt_j$).

\medskip\noindent\textbf{Notation.} Let $H$ be an induced subgraph of $G$. We define some terminology w.r.t. the graph $H = G[A' \uplus B']$, where $A' \subseteq A$, and $B' \subseteq B$. For a vertex $v \in A'$, and a color $j$, let $N^H_j(v) \subseteq B' \cap B_j$ denote the set of neighbors of $v$ of color $j$. More formally, let $N^H_j(u) \coloneqq \{ \fs \in B' : f(\fs) = j \text{ and } (u, \fs) \in E(H) \}$. Furthermore, for a subset $S \subseteq A$, let $N_j(S) \coloneqq \bigcup_{v \in S} N_j(v)$. We also define $d^H_j(v) \coloneqq |N^H_j(v)|$. We call $d^H_j(v)$ as $j$-degree of $v$ in the graph $H$. We may omit the superscript from these notations when $H = G$.  Finally, for a vertex $v \in A'$, we define $H \bbslash v$ as the graph obtained by deleting $N^H[v]$, i.e., $v$ and all of its neighbors in $B'$.  

Consider an instance $\cI = (G, [r], f, t, k)$ of \pbds. For a vertex $\fs \in B$, $f_{\setminus \fs}$ is a restriction of $f$ to the set $B \setminus \{\fs\}$. %We say that a subset $S \subseteq A$ has \emph{coverage vector} $(t'_1, t'_2, \ldots, t'_r)$ w.r.t., $\cI$ if and only if, for each color $j$, $|N^G_j(S)| = t'_j$. 

%, i.e., $f': B \setminus \{\fs\}$ such that $f'(\fp) = f(\fp)$ for all $\fp \in B \setminus \{\fs\}$. 
%The operation of \emph{deleting a color} $j \in [r]$ results in a new instance of \pbds, after the following operations: (1) taking appropriate restriction of $f$ to the vertices of $B$ of the remaining colors $[r] \setminus \LR{j}$, (2) taking the restriction of $t$ to $[r] \setminus \LR{j}$, and (3) deleting the vertices of color $j$ from $B$.

\noindent \textbf{Matroids and representative sets:} A matroid $\mathcal{M}=(U,I)$, consists of a finite universe $U$ and a family $I$ of sets over
$U$ that satisfies the following three properties:
\begin{enumerate}
	\item $\emptyset \in I$,
	
	\item if $A \in I$ and $B \subseteq A$, then $B \in I$,
	
	\item  if $A \in I$ and $B \in I$ and $|A| < |B|$ then there is an element $b \in B \setminus A$
	such that $A \cup \{b\} \in I$.
	
\end{enumerate}

Each set $S \in I$ is called an independent set and an inclusion-wise maximal independent set is known as a basis of the matroid. The rank of a set $P \subseteq U$ is defined as the size of the largest subset of $P$ that is independent, and is denoted by $\rank(P)$. Note that $\rank(P) \le \rank(Q)$ if $P \subseteq Q$.
From the third property of matroids, it is easy to observe that every inclusion-wise maximal independent set has the same size; this size is referred to as the \emph{rank} of the matroid. For any $P \subseteq U$, define the closure of $P$, $\cl(P) \coloneqq \LR{x \in U: \rank(P \cup \LR{x}) = \rank(P)}$. Note that $\cl(P) \subseteq \cl(Q)$ if $P \subseteq Q$. 

A matroid is linear (representable) if it can be defined using linear independence, that is for any such matroid $\mathcal{M}=(U,I)$, one can assign to every $e \in U$ a vector $v_e$ over some field (where the different elements of the universe should all be over the same field $\mathbf{F}$ and have the same dimension) such that a set $S \subseteq U$ is in $I$  if and only if the set $\{v_e: e \in  S\}$ forms a linearly independent set of vectors. 

For a matroid $\cM = (U, I)$ and an element $u \in U$, the matroid obtained by contracting $u$ is represented by $\cM' = \cM/u$, where $\cM' = (U \setminus \LR{u}, I')$, where $I'= \LR{ S \subseteq U \setminus \LR{u} : S \cup \LR{u} \in I }$. Let $\cM = (U, I)$ be a matroid of rank $k$. Then, for any $0 \le k' \le k$, one can define a truncated version of the matroid as follows: $\cM_{k'} = (U, I_{k'})$, where $I_{k'} = \LR{S \in I: |S| \le k'}$. It is easy to verify that the contraction as well as truncation operation results in a matroid. Furthermore, given a linear representation of a matroid $\cM$, a linear representation of the matroid resulting from contraction/truncation can be computed in (randomized) polynomial time \cite{DBLP:conf/acid/Marx06,DBLP:conf/icalp/LokshtanovMPS15}. Alternatively, given an oracle access to the original matroid $\cM$, one can simulate the oracle access to the contracted/truncated matroid with at most a linear overhead.

Next, we state the following crucial definition of representative families. 

\begin{definition}[\cite{fomin2014efficient,DBLP:books/sp/CyganFKLMPPS15}] \label{repsetdefn}
	Let $\mathcal{M}$ be a matroid and $\mathcal{A}$ be a family of sets of size $p$
	in $\mathcal{M}$. A subfamily $\mathcal{A}' \subseteq \mathcal{A}$ is said to $q$-represent $\mathcal{A}$ if for every set $B$ of size
	$q$ such that there is an $A \in \mathcal{A}$ such that $A\cup B$ is an independent set, there is an $A' \in \mathcal{A}'$  such that $A'\cup B$ is an independent set. If  $\mathcal{A}'$
	$q$-represents $\mathcal{A}$, we write $\mathcal{A}' \subseteq^q_{rep} \mathcal{A}$.
\end{definition}

We call a family of sets of size $p$ as $p$-family.

\begin{proposition}[\cite{fomin2014efficient,DBLP:books/sp/CyganFKLMPPS15}] \label{prop:repset} There is an algorithm that, given a matrix $M$ over a field $GF(s)$, representing a matroid $\mathcal{M} = (U, F)$ of rank $k$, a $p$-family ${\cal A}$ of independent sets in $\mathcal{M}$, and an integer $q$ such that $p + q = k$, computes a $q$-representative family ${\cal A}'\subseteq ^q_{rep} \mathcal{A}$ of size at most
	$\binom{p+q}{p}$ using at most $\Oh(|{\cal A}|(\binom{p+q}{p})p^{\omega}+(\binom{p+q}{p})^{\omega-1})$
	%$\mathcal{O}(|A|(\binom{p+q}{p}p^{\omega} +(\binom{p+q}{p})^{\omega − 1}))$ 
	operations over $GF(s)$.
\end{proposition}

In the following lemma, we show that $\mathcal{A}' \subseteq^1_{rep} \mathcal{A}$ can be computed in polynomial time using oracle access to $\cM$. 

\begin{lemma} \label{lem:oraclerepset}
	Let $\cM= (U, I)$ be a matroid given via oracle access, and let $\mathcal{A}$ be a $1$-family of subsets of $U$. Then, $\mathcal{A}' \subseteq^{k-1}_{rep} \mathcal{A}$ can be computed in polynomial time.
\end{lemma}
\begin{proof}
	Let $A = \LR{x : \LR{x} \in \mathcal{A}}$ be the set of underlying elements corresponding to the sets of $\mathcal{A}$. We compute an inclusion-wise maximal subset $R \subseteq A$ that is independent in $\cM$. Note that this subset can be computed using $O(|U|k)$ queries, where $k = \rank(\cM)$. We claim that $\mathcal{A}' = \LR{\LR{y}: y \in R}$ is a $1$-representative set of $\mathcal{A}$.
	
	Consider any $u \in A$ and $X \subseteq U$ of size $k$ such that $u \in X$ and  $X \in \cI$. We will show that there exists some $u' \notin R$ such that $X \setminus \LR{u} \cup \LR{u'} \in \cI$. Since $R$ is an inclusion-wise maximal independent subset of $A$, it follows that $u \in \cl(R) \subseteq \cl(R \cup (X \setminus \LR{u}))$. This means that $\rank(R \cup (X \setminus \LR{u})) = \rank(R \cup (X \setminus \LR{u}) \cup \LR{u}) = \rank(R \cup X)$. However, $\rank(R \cup X) \ge \rank(X)$. Therefore, we obtain that $\rank(R \cup (X \setminus \LR{u})) \ge \rank(X)$. This implies that there exists some $u' \in R$ such that $X \setminus \LR{u} \cup \LR{u'} \in \cI$. 
\end{proof}

\section{Reduction from \mfmaxsat to \mfmaxcov} 
\label{sec:reduction}

In this section, we begin with a polynomial time approximate preserving randomized reduction from \fmaxsat to \fmaxcov. The success probability of the reduction is $\OO((\epsilon/r)^{k})$.  
Recall that in the \fmaxsat problem, given a \textsf{CNF}-formula $\Phi$ with $\chi : \cla{\Phi} \rightarrow [r]$, a coverage demand function $t : [r] \rightarrow \mathbb{N}$ and an integer $k$, the goal is to find an assignment of weight at most $k$ that satisfies at least $t(i)$ (also denoted as $t_i$) clauses of color class $i$ (an assignment $\Psi$ satisfying these properties is called   {\em  optimal weight $k$ assignment}). 
% \todo{move definitions to intro}

We begin with some basic definitions.
Let $\Phi$ be a \textsf{CNF}-formula. By $\var{\Phi}$ and $\cla{\Phi}$, we denote the set of variables and clauses in the formula $\Phi$, respectively. An assignment to a \cnf-formula $\Phi$ is a function $\Psi : \var{\Phi} \rightarrow \{0, 1\}$. The weight of an assignment is the number of variables that have been assigned $1$. By $T(\Psi)$ and $F(\Psi)$, we denote the set of variables assigned $1$ and $0$ by the assignment $\Psi$, respectively.
For a clause $c \in \cla{\Phi}$, $\var{c}$ is the set of variables that occur in the clause $c$ as a positive or negative literal. Similarly, for a set of clauses $C \in \cla{\Phi}$, $\var{C}$ is the set of variables that occur as a positive or negative literal in a clause $c \in C$. %For an $i \in [r]$, $\chi^{-1}(i)$ is the set of clauses in color class $i$. %The variable clause incidence graph of a \textsf{CNF}-formula $\Phi$ is a bipartite graph $G_{\Phi} = (A \uplus B, E)$ with a vertex for each variable in $\var{\Phi}$ in partition $A$ and a vertex for each clause in $\cla{\Phi}$ in partition $B$. There is an edge between a vertex $v \in A$ and a vertex $c \in B$ if the variable corresponding to $v$ occurs as a positive or negative literal in the clause corresponding to $c$, i.e., $v \in \var{C}$. 

%The algorithm based on reduction to \satcc (\Cref{alg:red}) 
Our reduction (\Cref{alg:red}) takes an instance $(\Phi, \chi, t, k)$ of \fmaxsat as input.  It constructs a random assignment $\Psi$ by setting each variable to $1$ with probability $p$ and $0$ with probability $1 - p$. It constructs a new formula by first removing the set of clauses that are satisfied negatively by $\Psi$, followed by removing negative literals from the remaining clauses. It reduces the formula to an instance $(\cU, \cF, \chi', t', k')$ of \fmaxcov as described in Step~\ref{red-max-cove} of \Cref{alg:red}. Clearly, this reduction takes polynomial time. 
%Figure~\ref{fig:reduction} illustrates the reduction.  
%If the $(\cU, \cF, \chi', t', k')$ is a \nin of \fmaxcov, the algorithm returns \no. Otherwise, it returns a weight $k$ assignment that satisfies at least $(1 - \epsilon)t_i$ clauses of each color class $i$.

% \begin{figure}[h]
	%    \centering
	%    \includegraphics[width=12cm,height=6cm]{reduction.png}
	%    \caption{Illustration of reduction from \fmaxsat to \fmaxcov}
	%    \label{fig:reduction}
	%\end{figure} 
	
	%We state the reduction procedure in \Cref{alg:red}.
	
	\begin{algorithm}[t]
		\caption{Reduction Algorithm$(\cI= (\Phi, \chi, t, k)$ of \fmaxsat $)$} \label{alg:red}
		\begin{algorithmic}[1]
			\Statex 
			\State Construct a random assignment $\Psi$ as follows. For each variable $x \in \var{\Phi}$, independently set $\Psi(x)$ to $1$ with probability $p$ and $0$ with probability $1 - p$. \bluecomment{We will later set $p = \frac{ \epsilon}{2r}$.}
			\State Construct a new formula $\Phi'$ as follows:
			\begin{itemize} 
				\item Let $N \subseteq \cla{\Phi}$ be the set of clauses that are satisfied negatively by $\Psi$. Then, $\cla{\Phi'} = \cla{\Phi} \setminus N$.
				\item For each $c \in \cla{\Phi'}$, remove all the variables in $c$ that occur either as a negative literal or set to $0$ by $\Psi$.
				\item For each $c \in \cla{\Phi'}$, add $\var{c}$ to $\var{\Phi'}$.
			\end{itemize}
			\State Construct an instance ${\cal J}_{\Psi}=(\cU, \cF, \chi', t', k')$ of \fmaxcov as follows: \label{red-max-cove}
			\begin{itemize}
				\item Set $\cU = \cla{\Phi'}$.
				\item For each $v \in \var{\Phi'}$, add a set $f_v$ to $\cF$ where $f_v = \{c \in \cla{\Phi'} \colon v \in \var{c}\}$. 
				\item For each $c \in \cla{\Phi'}$, if the corresponding element in $\cU$ is $e_c$, set $\chi'(e_c) = \chi(c)$.
				\item Set $t'(i) = t(i) - \frac{|N \cap \chi^{-1}(i)|}{1-\epsilon}$ for each $i \in [r]$. %\bluecomment{some clauses in color class $i$ might be already satisfied negatively.}
				\item Set $k' = k$.
			\end{itemize}
			%\State Return the instance $(\cU, \cF, \chi', t', k')$.
			%		\State Run the algorithm for \fmaxcov on $(\cU, \cF, \chi', t', k')$.
			%		\If{$(\cU, \cF, \chi', t', k')$ is a \nin of \fmaxcov}
			%		\State \Return \no.
			%		\Else{\text{ Let} $S$ be the solution for $(\cU, \cF, \chi', t', k')$}
			%		\For{each $x \in \var{\Phi}$}
			%		\State If $f_x \in S$, set $\sigma(x)$ to $1$, else set $\sigma(x)$ to $0$.
			%		\EndFor
			%		\EndIf
			%		\State \Return $\sigma$		
		\end{algorithmic}
	\end{algorithm}
	
	Next, we prove the correctness of our reduction. %Towards this, we prove the following two lemmas.
	For a \yin $\cI$ of \fmaxsat, let $\Psi^\star$ be an optimal weight $k$ assignment. 
	Let $N^\star$ be the set of clauses satisfied negatively by $\Psi^\star$, i.e., every clause in $N^\star$ contains a negative literal that is set to $1$, and let $P^\star$ be the set of clauses satisfied only positively by $\Psi^\star$, i.e., every clause in $P^\star$ contains a positive literal that is set to $1$ and no negative literal in this clause is set to $1$.
	By $N^\star_{i}$, we mean the set of clauses in color class $i$ satisfied negatively by $\Psi^\star$ and by $P^\star_{i}$, we mean the set of clauses in color class $i$ satisfied only positively by $\Psi^\star$.
	We call a random assignment, constructed in \Cref{alg:red}, \emph{good} if %denote our \emph{good} event $\mathcal{G}$ to be the case where 
	each variable in $T(\Psi^\star)$ (positive variables under $\Psi^\star$) is assigned $1$ by $\Psi$, i.e., $T(\Psi^\star) \subseteq T(\Psi)$, which occurs with probability at least $p^k$. For a good assignment $\Psi$, let $N_{i}$ denote the set of clauses in color class $i$ satisfied negatively by $\Psi$ and $P_{i}$ denote the set of clauses in color class $i$ satisfied only positively by $\Psi$. We say that an event $\mathcal{G}$ is \emph{good} if a good assignment $\Psi$ is generated in \Cref{alg:red}.   We begin with the following claim. 
	
	\begin{claim} \label{clm:satneg}
		Given a \yin $\cI$ of \fmaxsat, with probability at least $1/2$, a good assignment $\Psi$ satisfies at least $(1 - \epsilon)|N_i^\star|$ clauses negatively, for each $i\in [r]$.
	\end{claim}
	
	\begin{proof}
		Let $(\Phi, \chi, t, k)$ be a \yin of \fmaxsat. Let $\Psi$ be a good assignment, which occurs with probability at least $p^k$. We show that $\Psi$ satisfies at least $(1 - \epsilon)|N^\star|$ clauses negatively, with probability at least $1/2$. 
		Let $X_i$ be the number of clauses in $N^\star_{i}$ that are satisfied negatively by $\Psi$. We define an indicator random variable $x_j$, for each $j \in [|N^\star_{i}|]$, as follows.
		\[
		x_j = \begin{cases} 
			1 & \text{clause } c_j \in N^\star_{i} \text{ is satisfied negatively by } \Psi \\
			0 & \text{otherwise} 
		\end{cases}
		\]
		\[
		{\sf Pr}(x_j | \mathcal{G}) = {\sf Pr}(\text{clause } c_j \in N^\star_{i} \text{ is satisfied negatively by } \Psi | \mathcal{G}) \geq (1 - p)
		\]
		\[
		\mathbb{E}[X_i | \mathcal{G}] = \sum_{j \in [|N^\star_{i}|]} x_j \times {\sf Pr}(x_j | \mathcal{G}) \geq (1 - p)|N^\star_{i}|
		\]
		Let $Y_i = |N^\star_{i}| - |N_{i}|$, where $N_{i}$ is the set of clauses satisfied negatively by $\Psi$. Note that,
		\[
		Y_i = |N^\star_{i}| - |N_{i}| \leq |N^\star_{i} \setminus N_{i}| = |N^\star_{i}| - X_i
		\]
		Thus,
		\[
		\mathbb{E}[Y_i | \mathcal{G}] \leq |N^\star_{i}| - \mathbb{E}[X_i | \mathcal{G}] \leq p|N^\star_{i}|
		\]
		Since $Y_i \geq 0$, we can use Markov's inequality and get
		\[
		{\sf Pr}(Y_i \geq 2rp|N^\star_{i}| | \mathcal{G}) \leq \frac{\mathbb{E}[Y_i]}{2rp|N^\star_{i}|} \leq \frac{1}{2r}
		\]
		Since $Y_i = |N^\star_{i}| - |N_{i}|$, we get
		\[
		{\sf Pr}(|N_{i}| \leq (1 - 2rp)|N^\star_{i}| | \mathcal{G}) \leq \frac{1}{2r}
		\]
		By union bound,
		\[
		{\sf Pr} \lr{ \exists i \in [r], |N_{i}| \leq (1 - 2rp)|N^\star_{i}| | \mathcal{G} } \leq \sum_{i \in [r]} {\sf Pr}(|N_{i}| \leq (1 - 2rp)|N^\star_{i}| | \mathcal{G}) \leq \frac{1}{2}
		%Pr(\exists i \in [r], |N_{i}| \leq (1 - 2rp)|N^\star_{i}| | \mathcal{G}) \leq \sum_{i \in [r]} Pr(|N_{i}| \leq (1 - 2rp)|N^\star_{i}| | \mathcal{G}) \leq \frac{1}{2}
		\]
		This implies that
		\[
		{\sf Pr} \lr{ \forall i \in [r], |N_{i}| > (1 - 2rp)|N^\star_{i}| | \mathcal{G} } \geq \frac{1}{2}
		\]
		Setting $p = \frac{ \epsilon}{2r}$ gives us the required result, i.e., with probability at least $1/2$, for all colors $i \in [r]$, $|N_{i}| > (1 - \epsilon)|N^\star_{i}|$.
		%
		%We know that $F(\Psi) \subseteq F(\sigma)$. This implies that the set of clauses satisfied negatively by $\Psi$ is a subset of the set of clauses satisfied negatively by $\sigma$, which proves the claim. 
	\end{proof}
	
	\begin{lemma}\label{lem:red-fwd}
		If $\cI= (\Phi, \chi, t, k)$  is a yes-instance of \fmaxsat, then with probability at least $(\frac{\epsilon}{2r})^k$, 
		%for a good assignment $\Psi$,  
		the reduced instance ${\cal J}_\Psi=(\cU, \cF, \chi', t', k')$ is a yes-instance of  \fmaxcov.
	\end{lemma}
	
	\begin{proof}
		Let $\cI$ be a \yin of \fmaxsat and let $\Psi^\star$ be an optimal weight $k$ assignment. Further, let $N^\star$ be the set of clauses satisfied negatively by $\Psi^\star$, i.e., every clause in $N^\star$ contains a negative literal that is set to $1$, and let $P^\star$ be the set of clauses satisfied only positively by $\Psi^\star$, i.e., every clause in $P^\star$ contains a positive literal that is set to $1$ and no negative literal in this clause is set to $1$. 
		Then, there exists a set $V_{P^\star} \subseteq \var{P^\star}$ of size at most $k$ that satisfies all the clauses in $P^\star$ positively, i.e., for each clause in $P^\star$, there is a variable in $V_{P^\star}$ that occurs as a positive literal in it and is assigned $1$ under $\Psi^\star$. %Let $A = P^\star \cap N$, i.e., $A$ is the set of clauses in $P^\star$ that is also satisfied negatively by $\Psi$.
		%Let $V_{P^\star} \subseteq \var{P^\star}$ be the set of variables that satisfy the set of clauses in $P^\star$ positively, i.e., for each clause in $P^\star$, there is a variable in $V_{P^\star}$ that occurs as a positive literal in it and is assigned $1$ under $\Psi^\star$. 
		%We construct a solution for ${\cal J}_\Psi$ as follows. 
		% Let $P_i$ be the set of clauses in color class $i$ that are satisfied only positively by $\Psi$.
		Let $\Psi$ be a good assignment which is generated with probability at least $(\frac{\epsilon}{2r})^k$. 
		Since $\Psi$ is a good assignment, $T(\Psi^\star)\subseteq T(\Psi)$ and $F(\Psi)\subseteq F(\Psi^\star)$. %$V_{P^\star} \subseteq T(\Psi)$. Furthermore, 
		Hence, $P^\star \subseteq P$. Thus, $P^\star \subseteq \cU$ and for each variable in $\var{P^\star}$, we have a set in the family $\cF$. % for each $i\in [r]$, $P_i^\star \subseteq P_i$. 
		Let $Z=\{f_v\in \cF \colon v\in V_{P^\star}\}$. We claim that $Z$ is a solution to ${\cal J}_\Psi$. Clearly, $Z$ covers at least $|P_i^\star|$ elements in $\cU$, for each $i\in [r]$. We claim that $t_i'\leq |P_i^\star|$.  Since $\Psi^\star$ is a solution to $\cI$, it satisfies $t_i$ clauses for each $i\in [r]$. Since $t_i=|P_i^\star|+|N_i^\star|$, due to \Cref{clm:satneg}, we know that $t_i\leq |P_i^\star|+\frac{|N_i|}{1-\epsilon}$. Thus,  $|P_i^\star| \geq t_i - \frac{|N_i|}{1-\epsilon}$. Since $t'_i= t_i - \frac{|N_i|}{1-\epsilon}$, $t_i'\leq |P_i^\star|$. This completes the proof. %Since 
		%Thus, there exists a subset of $T(\Psi)$ of weight at most $k$ that satisfies at least $|P_i^\star|$ clauses in $P_i$ for each $i\in [r]$.  %all the clauses in $P^\star \setminus A$ as $V_{P^\star}$ satisfies the clauses in $P^\star \setminus A$ positively. 
	\end{proof}
	\begin{lemma}\label{lem:red-rev}
		Assume that  $\cI= (\Phi, \chi, t, k)$  is a yes-instance of \fmaxsat. If there exists  $(1-\epsilon)$-approximate solution for ${\cal J}_\Psi=(\cU, \cF, \chi', t', k')$, where $\Psi$ is a good assignment, then there exists  $(1-\epsilon)$-approximate solution for  $\cI$ with probability at least $\frac{1}{2}$. 
		% an assignment $\sigma$ of weight $k$ such that it satisfies at least $(1 - \epsilon)t_i$ clauses of color class $i$, for each $i \in [r]$.
	\end{lemma}
	
	\begin{proof}
		Let $\Psi^\star$ be an optimal assignment. Due to \Cref{clm:satneg},  $\Psi$ satisfies at least $(1 - \epsilon)|N_i^\star|$ clauses negatively, for each $i\in [r]$, with probability at least $1/2$. Let $S$ be a $(1-\epsilon)$-approximate solution to ${\cal J}_\Psi$. We construct an assignment $\sigma$ as follows: if $f_x\in S$, then $\sigma(x)=1$, otherwise $0$. We claim that $\sigma$ is a $(1-\epsilon)$-approximate solution to $\cI$. Due to the construction of ${\cal J}_\Psi$, note that $\cF$ does not contain a set corresponding to the variable that is set to $0$ by $\Psi$. Thus, if $\Psi(x)=0$, then $\sigma(x)=0$. Hence, $\sigma$ satisfies at least $(1 - \epsilon)|N_i^\star|$ clauses negatively, for each $i\in [r]$, with probability at least $1/2$. Next, we argue that $\sigma$ satisfies at least $(1 - \epsilon)|P_i^\star|$ clauses only positively, for each $i\in [r]$. Since $S$ is a $(1-\epsilon)$-approximate solution to ${\cal J}_\Psi$, for each $i\in [r]$, $S$ covers at least $(1-\epsilon)t_i'$ elements. %Let $P_i$ be the set of clauses in color class $i$ that are satisfied only positively by $\Psi$. 
		Recall that $\cU$ contains an element corresponding to each clause in $\cup_{i\in [r]}P_i$. Thus, $\sigma$ satisfies at least $(1-\epsilon)t_i'$ clauses only positively for each $i\in [r]$. Recall that $t_i'=t_i-\frac{|N_i|}{1-\epsilon}$. Thus, $|P_i|+|N_i| \geq (1-\epsilon)(t_i-\frac{|N_i|}{1-\epsilon})+|N_i| = (1-\epsilon)t_i$. Hence, $\sigma$ is a factor $(1-\epsilon)$-approximate solution for $\cI$. 
	\end{proof}
	Due to Lemma~\ref{lem:red-fwd} and ~\ref{lem:red-rev}, we have the following result. 
	%\begin{theorem}\label{thm:red-max-cov}
	%There exists a polynomial time randomized algorithm that given a \yin $\cI$ of \fmaxsat generates a \yin ${\cal J}$ of \fmaxcov with probability at least $(\frac{\epsilon}{2r})^k$. Furthermore, given a factor $(1-\epsilon)$-approximate solution of ${\cal J}$, we can construct a $(1-\epsilon)$-approximate solution of ${\cal I}$ with probability at least $1/2$. 
	%\end{theorem}
	
	\begin{theorem}\label{thm:red-max-cov}
		There exists a polynomial time randomized algorithm that given a \yin $\cI$ of \fmaxsat generates a \yin ${\cal J}$ of \fmaxcov with probability at least $(\frac{\epsilon}{2r})^k$. Furthermore, given a factor $(1-\epsilon)$-approximate solution of ${\cal J}$, it can be extended to a $(1-\epsilon)$-approximate solution of ${\cal I}$ with probability at least $1/2$. 
	\end{theorem}

	Note that if the variable-clause incidence graph of the input formula belongs to a subgraph closed family $\mathcal{H}$, then the incidence graph of the resulting instance of \fmaxcov will also belong to $\mathcal{H}$. Thus, due to \Cref{thm:red-max-cov} and Theorem 1.2 in \cite{DBLP:conf/soda/0001KPSS0U23}, we have the following result, which is an improvement over Theorem 1.1 in \cite{DBLP:conf/soda/0001KPSS0U23}. %if we have an \FPTAS for \fmaxcov where the incidence graph of the input instance satisfies a certain property, then we also obtain an \FPTAS for \fmaxsat where the incidence graph of the input instance also satisfies the same property.  
	
	\begin{theorem} \label{thm: corrkdd}
		There is a randomized algorithm that given a \yin $\cI$ of {\sc CC-Max-SAT}, where the variable-clause incidence graph is $K_{d,d}$-free, returns a factor $(1-\epsilon)$-approximate solution with probability at least $(1-\frac{1}{e})$, and runs in time $(\frac{rdk}{\epsilon})^{\OO(dk)}(n+m)^{\OO(1)}$. 
	\end{theorem}
	
	The $(\frac{r}{\epsilon})^{\OO(k)}$ factor in the running time of the algorithm in Theorem~\ref{thm: corrkdd} comes by repeating the algorithm in Theorem~\ref{thm:red-max-cov}, followed by Theorem 1.2 in \cite{DBLP:conf/soda/0001KPSS0U23}, independently $(\frac{r}{\epsilon})^{\OO(k)}$ many times. This also boosts the success probabilty to at least $(1-\frac{1}{e})$. 
	
	\begin{remark}
		Note that the reduction from \fmaxsat to \fmaxcov also works in presence of matroid constraint(s) on the set of variables assigned $1$. Recall that in the former (resp. latter) problem, we are given a matroid $\cM$ on the set of variables (resp. sets), and the set of at most $k$ variables assigned $1$ (resp. at most $k$ sets chosen in the solution) is required to be an independent set in $\cM$. This follows from the fact that the randomized algorithm preserves the optimal independent set in the set cover instance with good probability. 
	\end{remark}

\section{An \textsf{FPT-AS} for \pbds with Bounded Frequency} 
\label{sec:bounded-degree}

We first design the \textsc{Bucketing} subroutine in \Cref{subsec:bucketing}. Then, in \Cref{subsec:fptas-freq-d}, we design and analyze the \FPTAS for \pbds  when each vertex in $B$ has degree at most $d$ (i.e., frequency at most $d$). 

\subsection{The \bucketing Subroutine.}\label{subsec:bucketing}
In this section, we design a subroutine, called \bucketing, which (or slight variations thereof) will be crucially used to design \FPTAS in the subsequent sections. As mentioned in the introduction, we divide the set of vertices of $A$ into a bounded number of equivalence classes depending on the size of their $j$-neighborhoods. Fix a color $j$. At a high level, two vertices $u$ and $v$ belonging to a particular equivalence class will have $j$-degrees that are \emph{approximately equal}. There are two exceptions to this: (A) For a color $j$, all vertices of degree at least $t_j$ are treated as equivalent as far as color $j$ is concerned (i.e., we still classify such vertices based on their $j'$-degrees for other colors $j'$) -- since any single vertex from such a class is sufficient to entirely take care of color $j$. (B) all vertices of $j$-degree less than $\frac{\epsilon t_j}{2k}$ are treated as equivalent as far as color $j$ is concerned, for the following reason. The difference between the $j$-degrees between such two vertices is at most $\frac{\epsilon t_j}{2k}$. Thus, even if we make a bad choice at most $k$ times, we only lose at most $\epsilon t_j$ coverage for the color $j$. Thus, the \emph{interesting range} of $j$-degrees is between $\left[\frac{\epsilon t_j}{2k}, t_j\right]$, which is sub-divided into intervals of range $(1+\epsilon)$. It is easy to see that the vertices are partitioned into $\Oh(\log_{1+\epsilon}k)$ classes for a particular color. We proceed in a similar manner for each color $j \in [r]$, and then we obtain our final set of equivalence classes, such that all the vertices in a particular class are \emph{equivalent} w.r.t.\ each color in the sense described above. 
%We note that a variation of this \textsc{Bucketing} subroutine will also be used in the subsequent section, where we consider the more general case of $K_{d, d}$-free graph $G$.

%At a high level, this procedure groups the vertices of $A$ into ``bags'', such that, all the vertices belonging to a particular bag have \emph{approximately equal} degree (i.e., within a factor of $1+\epsilon)$ w.r.t.\ \emph{all} colors $j \in [r]$. Thus, in isolation, any two vertices $u$ and $v$ from a particular bag are interchangeable, since we tolerate an $\epsilon$-factor loss in the coverage. However, consider a set $S \subseteq A$ of size at least $2$. Consider a vertex $u \in S$, and a vertex $v$ that belongs to the same bag as $u$. In this case, however, $u$ and $v$ may not be interchangeable w.r.t.\ $S$. That is, the coverage vectors of $S$ and $S \setminus \{u\} \cup \{v\}$ may be very different, if, e.g., $v$ has shares many common neighbors with $S \setminus \{u\}$. Nevertheless, we are able to carefully design approximation algorithms that can leverage this approximate interchangeability in certain situations.

Our algorithm for \pbds is recursive, and will use the \textsc{Bucketing} procedure as a subroutine. During the course of the recursive algorithm, we may modify the instance in a variety of ways -- delete a subset of vertices, restrict the coverage function $t$, decrement the value of $k$, or delete a subset of colors. However, in the \bucketing subroutine we require the \emph{original} value of $k$, which we will denote by $k^*$. Now, we formally state the procedure.

\medskip\noindent\bucketing procedure.
\\Let $\cI = (G = (A \uplus B, E), [r], f, t, k)$ be the \emph{current} instance of \pbds. Then, let $\lambda \coloneqq \lceil \log_{(1+\epsilon)}\frac{2k^*}{\epsilon} \rceil$. Fix a color $j \in [r]$. For every $1 \le \alpha \le \lambda$, we define $$A(j, \alpha) \coloneqq \LR{ v \in A : \frac{t_j}{(1+\epsilon)^{\alpha}} \le d^j(v) < \frac{t_j}{(1+\epsilon)^{\alpha-1}}  }.$$
We also define $A(j, 0) \coloneqq \LR{ v \in A : d_j(v) \ge t_j }$, and $A(j, \lambda+1) \coloneqq \LR{v \in A: d_j(v) < \frac{t_j}{(1+\epsilon)^{\lambda}} }$. 

Let $\mathbf{V} = \{0, 1, \ldots, \lambda, \lambda+1\}^r$, and consider an arbitrary vector $\mathbf{v} \in \mathbf{V}$. Let $\mathbf{v} = (\alpha_1, \alpha_2, \ldots, \alpha_r)$. Then, we define $A(\mathbf{v}) \coloneqq \bigcap_{j = 1}^r A(j, \alpha_j)$. We call any such $A(\mathbf{v})$ as a \emph{bag}. The \bucketing procedure first computes the set of bags as defined above, and returns only the set of non-empty bags, which form a partition of $A$. It is easy to see that the procedure can be implemented in polynomial time.

\begin{observation} \label{obs:buckets}
	The items 1-3 in the following hold for any color $j \in [r]$.
	\begin{enumerate}
		\item For any $v \in A(j, \lambda+1)$, $d_j(v) < \frac{\epsilon t_j}{2k^*}$.
		\item For any $v \in A(j, 0)$, $d_j(v) \ge t_j$.
		\item For any $1 \le \alpha \le \lambda$, and for any $u, v \in A(j, \alpha)$, it holds that $\frac{d_j(u)}{(1+\epsilon)} \le d_j(v) \le (1+\epsilon) \cdot d_j(u)$.
	\end{enumerate}
	The number of bags returned by \bucketing is  bounded by $ \prod_{j = 1}^r (\lambda+2) \le \lr{\frac{6 \log k^*}{\epsilon^2}}^r =: L$.
\end{observation}

\subsection{The Algorithm} \label{subsec:fptas-freq-d}

Our algorithm for \pbds, when each vertex in $B$ has degree at most $d$ is given in \Cref{alg:freqd}. This algorithm is recursive, and takes as input an instance $\cI$ of \pbds; initially the algorithm is called on the original instance. We assume that we remember the value of $k^*$ from the original instance, since that is used in the \textsc{Bucketing} subroutine. Now we discuss the algorithm. In the algorithm, first we use the \textsc{Bucketing} subroutine to partition the vertices of $A$ into a number of equivalence classes. Then, we pick one of the bags uniformly at random -- essentially, in this step we are trying to guess the bag that has a non-empty intersection with a hypothetical solution $S$. \footnote{We note that this step can be replaced by deterministically branching on each of the bags instead; but since the next step is inherently randomized, we continue with the current presentation for the simplicity of exposition.} 

Suppose we correctly guess such a bag $A(\mathbf{v})$. Then, we arbitrarily pick a vertex $v$ from this bag, and use it to define a probability distribution on all the vertices in $A$. This distribution places a constant ($\ge 1/2$) probability mass on sampling $v$, and the rest of the probability mass is split proportional to the number of common neighbors of a vertex $w$ with $v$. Then, we sample a vertex $u$ according to this distribution, add it to the solution, and recurse on the residual instance. Here, the intuition is that, if the vertices in $S$ have a lot of common neighbors with $v$, then one of the vertices from $S$ will be sampled with reasonably large probability. In this case, we recurse on a vertex from a hypothetical solution, i.e., a ``correct choice''. Otherwise, if the vertices in $S$ have very few common neighbors with $v$, then we claim that we can replace a vertex in $S \cap A(\mathbf{v})$ with $v$, and still obtain a good (hypothetical) solution to compare against. In this case, we argue that $v$ is the ``correct choice''. Now, we state the algorithm formally, and then proceed to the analysis.

\begin{algorithm}[t]
	\caption{\textsc{\pbds}$(\cI= (G = (A \uplus B, E), [r], f, t, k))$} \label{alg:freqd}
	\begin{algorithmic}[1]
		\Statex 
		\If{$k = 0$} 
		\State \Return $\emptyset$ 
		\EndIf
		%\Statex {\small \texttt{All quantities in the following are computed w.r.t. $G$, but the subscript is omitted for simplicity.}}
		\State Let $\cal A$ be the set of bags returned by applying \bucketing on $\cI$
		\State Choose a bag $A(\mathbf{v}) \in {\cal A}$ uniformly at random, and select an arbitrary vertex $v \in A(\mathbf{v})$
		\State Define the following quantities w.r.t. the vertex $v$: 
		\State \qquad $\bullet$ For any $w \in A \setminus {v}$ and any color $j$, let $h_j(w) \coloneqq \frac{|N^G_j(w) \cap N^G_j(v)|}{|N^G_j(v)|}$, 
		\State \qquad $\bullet$ Let $p(w) \coloneqq \frac{1}{2rd} \cdot \sum_{j = 1}^r h_j(w)$
		\State \qquad $\bullet$ Let $p(v) \coloneqq \frac{1}{2}$
		\State Sample a vertex $u \in V(G)$ at random proportional to the quantities $p(\cdot)$
		\State $\cI' \gets \textsc{PruneInstance}(\cI, u)$ \bluecomment{Residual instance after adding $u$ to the solution}
		\State \Return $\tilde{S} \cup \{u\}$, where $\tilde{S} \gets \textsc{\pbds}(\cI')$ 
		%\Statex \Comment{\small \texttt{$f'$ is a restriction of $f$ on the remaining vertices in $G\bbslash u$, and $\tau'$ is the \emph{residual coverage requirement}, i.e., $\tau'(j) = \max\LR{\tau(j) - |N_j(u)|, 0}$}}
	\end{algorithmic}
\end{algorithm}
\begin{algorithm*}{\textbf{procedure} \textsc{PruneInstance}$(\cI, u)$ \qquad $\cI = (G = (A \uplus B, E), [r], f, t, k)$ and $u \in A$}
	\hrule 
	
	\begin{algorithmic}[1]
		\Statex
		\State Let $J \coloneqq \LR{ j \in [r] : |N_j(u)| \ge t_j }$ \bluecomment{$J$ is the set of colors already satisfied by $u$}
		\If {$J = \emptyset$} 
		\State $G' \coloneqq G \bbslash u$ \bluecomment{Remove $u$ from $A$ and all of its neighbors from $B$}
		\State $f'$ is the restriction of $f$ to $B \setminus \bigcup_{j \in [r]} N_j(u)$
		\State $t'_j = t_j - |N_j(u)|$ for all colors $j$ \bluecomment{Account for addition of $u$ into the solution}
		\State $\cI' = (G', [r], f', t', k-1)$
		\Else
		\State $r' \coloneqq r - |J|$, and w.l.o.g. assume that $J = \LR{r'+1, r'+2, \ldots, r}$
		\State $G' = G[A' \uplus B']$, where $A' \coloneqq A \setminus \{u\}$, and $B' \coloneqq B \setminus \lr{N(u) \cup \bigcup_{j \in J} f^{-1}(j) }$
		\Statex \bluecomment{Remove $u$, its neighbors in $B$, and vertices of all the colors satisfied by $u$}
		\State $f'$ is the restriction of $f$ to $B'$, and $t'_j = t_j - |N_j(u)|$ for $j \in [r'] $
		\State $\cI' = (G', [r'], f', t', k-1)$
		\EndIf
		\State \Return $\cI'$
	\end{algorithmic}
\end{algorithm*}

First we state the following observation.

\begin{observation}\label{obs:prob-distribution}
	$\sum_{w \in A} p(w) \le 1$.
\end{observation}
\begin{proof}
	If $v$ is the vertex chosen in line 5 of the algorithm, then $p(v) \coloneqq 1/2$. Therefore, we show that $\sum_{w \in A_1} p(w) \le \frac{1}{2}$, where $A_1 = A \setminus \LR{v}$.
	\begin{align*}
		\sum_{w \in A_1 } p(w) &= \frac{1}{2rd}\sum_{w \in A_1} \sum_{j = 1}^{r} h_j(w)
		\\&= \frac{1}{2rd}  \sum_{j = 1}^r \frac{1}{|N^G_j(v)|} \sum_{w \in A_1} |N^G_j(w) \cap N^G_j(v)| 
		%\le \frac{1}{2rd} \sum_{j = 1}^r \frac{1}{|N^G_j(v)|} \cdot d \cdot |N^G_j(v)| = \frac{1}{2}
		\\&\le \frac{1}{2rd} \sum_{j = 1}^r \frac{1}{|N^G_j(v)|} \cdot d \cdot |N^G_j(v)| = \frac{1}{2}
		%\\&= \frac{1}{2}
	\end{align*}
	Where, the second-last inequality follows from the fact that every vertex in $N^G_j(v)$ is counted at most $d$ times in %the sum 
	$\sum_{w \in A_1} |N^G_j(w) \cap N_j(v)|$, since every vertex in $B$ has degree at most $d$.
\end{proof}

Now we explain how to sample a vertex proportional to the quantities $p(\cdot)$. Note that \Cref{obs:prob-distribution} implies that $\ell \coloneqq \sum_{w \in A} p(w) \le 1$. Also note that $\ell \ge p(v) = 1/2$. Now, we sample a vertex from $A$ such that the probability of sampling a vertex $w$ is equal to $p(w)/\ell$ -- this can be done, e.g., by mapping the vertices to disjoint sub-intervals of $[0, 1]$ of length $p(w)/\ell$, and then sampling from uniform distribution over $[0, 1]$. Note that the sum of probabilities is equal to $1$, and thus this is a valid probability distribution. Finally, observe that for any set $W \subseteq A$ of vertices, the probability that a vertex from $W$ is sampled is equal to $\sum_{w \in W} p(w)/\ell \ge \sum_{w \in W}p(w)$. Next we prove the correctness of our algorithm. %Now we proceed towards proving the correctness of our algorithm.

\begin{lemma}\label{lem:random-induction}
	%Consider the instance $(G, f, \tau, k)$ at any point in the algorithm, and suppose there exists a set $S \subseteq A$ of size at most $k$ such that, for each color $j \in [r]$, $|N_j(S)| \ge t_j$. \\Then, with probability at least $q$, the algorithm returns a vertex-subset $S' \subseteq A$, such that for each $j \in [r]$, $|N_j(S')| \ge (1-2\epsilon)t_j$, where \[ q = \lr{\frac{1}{(6k/\epsilon^2)^r} \cdot \frac{\epsilon}{2rd} }^k. \]
	Consider a recursive call PCCDS$(\cI)$, where $\cI = (G = (A \uplus B), [r], f, t, k)$, and let $k^*$ be the value of $k$ from the \emph{original instance}.  Consider a set $S \subseteq A$ of size $k$ such that, for each $j \in [r]$, $|N^G_j(S)| = \ttt_j$. 
	
	Then, with probability at least $ \lr{ \frac{1}{L} \cdot \frac{\epsilon}{2rd} }^k$, the algorithm returns a subset $S' \subseteq A$ of size at most $k$, such that for each $j \in [r]$, $|N^G_j(S')| = t''_j$, where $$t''_j \ge (1-2\epsilon) \min(\ttt_j, t_j) - \frac{\epsilon k}{k^*} \cdot t_j.$$
\end{lemma}
\begin{proof}
	We prove this by induction. When $k = 0$, then there is nothing to prove. Now suppose the claim is true for $k-1 \ge 0$ and we prove it for $k$.
	
	Fix a set $S \subseteq V(G')$ of size $k$ such that $|N^G_j(S)| = \ttt_j$ for each $j \in [r]$. There exists a bag $A(\mathbf{v})$ such that $S \cap A(\mathbf{v}) \neq \emptyset$. Since in the first step (line 5), the algorithm picks a bag uniformly at random from at most $L$ bags, the probability that $A(\mathbf{v})$ is picked is at least $\frac{1}{L}$. We condition on the event that this choice is made, and proceed as follows. Suppose the algorithm picks $v \in A(\mathbf{v})$ in line 5. 
	We start with a relatively straightforward observation.
	
	\begin{observation} \label{obs:prune}
		Consider a vertex $u \in S$, and consider calling \textsc{PruneInstance}$(\cI, u)$. Let $J$ be the set as defined in line 1 in this procedure, and $\cI'$ be the instance that the call would return. Then, 
		\begin{enumerate}
			\item For any $j \in J$, $|N^{G}_j(u)| \ge t_j$, 
			\item For any $j \not\in J$, in the instance $\cI'$, it holds that $t'_j = t_j - |N^G_j(u)| \le t_j$, and
			\item For any $j \not\in J$, in the instance $\cI'$, it holds that $|N^{G'}_j(S \setminus \{u\})| = \ttt_j - |N^{G}_j(u)|$.
		\end{enumerate}
	\end{observation}
	\begin{proof}
		The first item follows from the definition of $N^G_j(u)$. The second item follows from the definition of $t'$. For the third item, we note that for any color $j \not\in J$, the vertices in $N^G_j(u)$ are removed from color $j$ in the instance $\cI'$. Thus, $N^{G'}_j(S \setminus \{u\})$ and $N^{G}_j(u)$ are disjoint sets in $V(G)$, and their sizes add up to $\ttt_j$.
	\end{proof}
	Next, we prove the following technical claim.
	
	\begin{claim} \label{cl:residual}
		Suppose the random vertex $u$ selected in line 10 belongs to the set $S$. Then, let $G'$ be the graph in the instance $\cI'$ obtained in line 11. Then, with probability at least $\lr{\frac{1}{L} \cdot \frac{\epsilon}{2rd}}^{k-1}$, it holds that for any $j \in [r]$, 
		\begin{equation}
			\left|N^G_j\lr{\tilde{S} \cup \{u\}}\right| \ge (1-2\epsilon) \cdot \min(\ttt_j, t_j) - \frac{\epsilon k}{k^*} t_j \label{eqn:claim-ineq}
		\end{equation}
		
	\end{claim}
	\begin{proof}
		Suppose the set $J$ as defined w.r.t. $u$ in line 1 of \textsc{PruneInstance} is non-empty. Then, for any $j \in J$, \Cref{obs:prune} implies that $|N^G_j(u)| \ge t_j$, which is at least the bound in the lemma. Thus, it suffices to focus on colors in $[r] \setminus J$.
		
		By \Cref{obs:prune} (item 3), the set $S \setminus \{u\}$ of size $k-1$ is such that for each $j \not\in J$, $|N^{G'}_j(S \setminus \LR{u})| = \ttt_j - |N^G_j(u)|$. Thus, by induction hypothesis (i.e., using \Cref{lem:random-induction}), with probability at least $\lr{\frac{1}{L} \cdot \frac{\epsilon}{2rd}}^{k-1}$, PCCDS$(\cI')$ returns a set $\tilde{S}$ of size at most $k-1$ that satisfies the following property: for any $j \not\in J$, 
		\begin{equation}
			\left|N^{G'}_j(\tilde{S})\right| \ge (1-2\epsilon) \cdot \min( \ttt_j - |N^{G}_j(u)|, \ t'_j ) - \frac{\epsilon (k-1)}{k^*} t'_j \label{eqn:lb-induction}
		\end{equation}
		Thus, it follows that the set $\tilde{S} \cup \{u\}$ satisfies that, for any $j \not\in J$,
		\begin{align*}
			\left|N^G_j(\tilde{S} \cup \{u\})\right| &=   \left|N^{G'}(\tilde{S})\right| + \left| N^G_j(u) \right|
			\\&\ge (1-2\epsilon) \cdot \lr{t_j - |N^G_j(u)|} - \frac{\epsilon (k-1)}{k^*} t'_j + |N^G_j(u)|
			\\&\ge (1-2\epsilon) \cdot t_j - \frac{\epsilon (k-1)}{k^*} t'_j
			%			\\&\ge (1-2\epsilon) \cdot \min(\ttt_j, t_j) - \frac{\epsilon k}{k^*} t'_j
			\\&\ge (1-2\epsilon) \cdot \min(\ttt_j, t_j) - \frac{\epsilon k}{k^*} t_j \tag{since $t'_j \le t_j$}
		\end{align*}
		This concludes the proof of the claim.
	\end{proof}
	
	Now we proceed with the inductive step, where we consider different cases.
	
	\textbf{Case 1.} $v \in S$. 
	\\In this case, the probability that the randomly chosen vertex in line 10 is equal to $v$, is at least $1/2$. We condition on this event, and using \Cref{cl:residual}, it follows that the set $\tilde{S} \cup \{v\}$ satisfies (\ref{eqn:claim-ineq}) for each $j \in [r]$, with (conditional) probability at least $\lr{\frac{1}{L} \cdot \frac{\epsilon}{2rd}}^{k-1}$. Thus, the unconditional probability that the set $\tilde{S} \cup \{u\}$ satisfies the required property is at least
	\begin{equation*}
		\frac{1}{2} \cdot \lr{ \frac{1}{L} \cdot \frac{\epsilon}{2rd} }^{k-1} \ge \lr{\frac{1}{L} \cdot \frac{\epsilon}{2rd}}^k.
	\end{equation*}
	
	%Then, notice that the residual instance $(G'  \bbslash v, k'-1 )$ admits a solution $S \setminus \{v\}$ that covers at least $(t'_j - d^j(v))$ edges of color $j$. Then, by inductive hypothesis, with probability at least $\lr{\frac{1}{(c \log n)^d} \cdot \frac{\epsilon}{2rd}}^{k-1}$ the algorithm returns a solution $\tilde{S}$ of size $k-1$ that covers at least $(1-2\epsilon) \cdot (t'_j - d^j(v))$ edges. Therefore, the number of edges of color $j$ covered by the solution $\tilde{S} \cup \{v\}$ is at least $(1-2\epsilon)(t'_j - d^j(v)) + d^j(v) \ge (1-2\epsilon)t'_j + 2\epsilon d^j(v) \ge (1-2\epsilon) t'_j$. Finally, the probability that the algorithm returns $\tilde{S} \cup \{v\}$ is at least 
	%$$\frac{1}{(2d \log n/\epsilon)^r} \cdot \frac{1}{2} \cdot \lr{\frac{1}{(2d \log n/\epsilon)^r} \cdot \frac{\epsilon}{2rd}}^{k-1} \ge \lr{\frac{1}{(2d \log n/\epsilon)^r} \cdot \frac{\epsilon}{2rd}}^k.$$
	
	\textbf{Case 2}. Suppose $v \not\in S$. Now we consider two sub-cases.
	
	\textbf{Case 2.1}: There exists a color $j$ such that $\sum_{w' \in S} |N^G_j(v) \cap N^G_j(w')| \ge \epsilon \cdot |N^G_j(v)|$. 
	\\We first claim that the probability that some vertex $w$ from the set $S$ is chosen to be $u$ in line 10 is at least $\frac{\epsilon}{r}$. Then, conditioned on this event, we will use the induction hypothesis to show the required bound.
	
	Fix a color $j$ satisfying the case assumption (if there are multiple such colors, pick one such color arbitrarily). Then, since color $j$ satisfies the case assumption, it follows that, %By definition, for any $w \in N_j(v) \cap S$,  $|N_j(w) \cap N_j(v)| \ge 1$, which implies that,
	\begin{align*}
		\sum_{w' \in S} h_j(v) &= \sum_{w' \in S} \frac{|N^G_j(w') \cap N^G_j(v)|}{|N^G_j(v)|} \ge \frac{1}{|N^G_j(v)|} \cdot \epsilon \cdot |N^G_j(v)| = \epsilon 
	\end{align*}
	Therefore,
	\begin{align*}
		\sum_{w' \in S} p(w') = \sum_{w' \in S} \frac{h(w')}{2rd} \ge \frac{1}{2rd} \cdot \sum_{w \in S} h_j(w') \ge \frac{1}{2rd} \cdot \epsilon
	\end{align*}
	Thus, the probability that some $w \in S$ will be chosen as the vertex $u$ is at least $\frac{\epsilon}{2rd}$. Now, we condition on this event. Then, by using \Cref{cl:residual} and an argument similar to case 1, the set $\tilde{S} \cup \{w\}$ satisfies (\ref{eqn:claim-ineq}) for each $j \in [r]$, with probability at least 
	\begin{equation*}
		\frac{1}{L} \cdot \frac{\epsilon}{2rd} \cdot \lr{\frac{1}{L} \cdot \frac{\epsilon}{2rd}}^{k-1} = \lr{\frac{1}{L} \cdot \frac{\epsilon}{2rd}}^k.
	\end{equation*}
	
	\textbf{Case 2.2}: For all colors $j$, %it holds that 
	$\sum_{w' \in S} |N_j(v) \cap N_j(w')| \le \epsilon \cdot |N_j(v)|$. 
	Hence, %it follows that 
	\begin{equation}
		\bigg|\lr{\bigcup_{w' \in S} N_j(w')} \cap N_j(v)\bigg| \le \sum_{w' \in S} |N_j(v) \cap N_j(w')| \le \epsilon \cdot |N_j(v)|. \label{ineq:small-intersection}
	\end{equation}
	Recall that the probability that the vertex $u$ is equal to $v$ is at least $1/2$, and we condition on this choice. Furthermore, recall that we have conditioned on the event that $A(\mathbf{v}) \cap S \neq \emptyset$, but due to case assumption $v \not\in S$. Therefore, there must exist a vertex $w \in A(\mathbf{v}) \cap S$. 
	
	In this case, we aim to show that $v$ ``approximately plays the role'' of $w \in A(\mathbf{v}) \cap S$ in our solution.	To this end, we consider different cases for the value of $\alpha_j$ in the vector $\mathbf{v}$. In each of the cases, we condition on the probability that the recursive call returns a set $\tilde{S}$ with the desired properties. Using induction, this happens with probability at least $\lr{\frac{1}{L} \cdot \frac{\epsilon}{2rd}}^{k-1}$. Thus, the unconditional probability is at least $ \lr{\frac{1}{L} \cdot \frac{\epsilon}{2rd}}^{k}$ as in Case 1. Now we proceed to the analysis of each of the cases, conditioned on the good events.
	
	\textbf{Case A: $\alpha_j = 0$.} 
	Since $v \in A(j, 0)$, %it holds that 
	$|N^G_j(v)| \ge t_j$. Thus, $|N^G_j(\tilde{S} \cup \{v\}| \ge t_j$, which is at least the claimed bound.
	
	\textbf{Case B: $\alpha_j = \lambda+1$.} 
	Since $v, w \in A(j, \lambda+1)$, %this implies that 
	$|N^G_j(v)| \le \frac{\epsilon t_j}{2k^*}$, and $|N^G_j(w)| \le \frac{\epsilon t_j}{2k^*}$.
	\\Now we analyze $\left|N^{G'}(S \setminus \{w\})\right|$. By case assumption, it follows that,
	\begin{align}
		\left|N^{G'}_j(S \setminus \{w\})\right| &= \ttt_j - |N^G_j(w)| - \bigg|\lr{\bigcup_{w' \in S} N^G_j(w')} \cap N^G_j(v)\bigg| \label{eqn:ineq-lb}
		\\&\ge \ttt_j - |N^G_j(w)| - |N^G_j(v)| \nonumber
		\\&\ge \ttt_j - \frac{2\epsilon t_j}{2k^*} = \ttt_j - \frac{\epsilon t_j}{k^*}. \nonumber
	\end{align}
	Then, by inductive hypothesis, the solution $\tilde{S}$ satisfies the first inequality in the following.
	\begin{align*}
		\left|N^{G'}_j(\tilde{S})\right| &\ge (1-2\epsilon) \cdot \lr{\ttt_j - \frac{\epsilon t_j}{k^*}} - \frac{\epsilon(k-1)t'_j}{k^*}
		\\&\ge (1-2\epsilon) \cdot \ttt_j - \frac{\epsilon t_j}{k^*} - \frac{\epsilon(k-1) t'_j}{k^*}
		\\&\ge (1-2\epsilon) \cdot \ttt_j - \frac{\epsilon k}{k^*} t_j \tag{Since $t'_j \le t_j$}
	\end{align*}
	which is at least the claimed bound.
	
	\textbf{Case C. $1 \le \alpha_j \le \lambda$.}
	\\Since $v, w \in A(j, \alpha_j)$, \Cref{obs:buckets} implies that
	\begin{equation}
		|N^G_j(w)| \le (1+\epsilon) \cdot |N^G_j(v)| \label{ineq:uw-approx}
	\end{equation}	
	Analogous to (\ref{eqn:ineq-lb}), we have the following.
	\begin{align*}
		\left|N^{G'}_j(S \setminus \{w\})\right| &= \ttt_j - |N^G_j(w)| - \bigg|\lr{\bigcup_{w' \in S} N^{G}_j(w')} \cap N^G_j(v)\bigg|
		\\&\ge \ttt_j - |N^G_j(w)| - \epsilon \cdot |N^G_j(v)| \tag{From (\ref{ineq:small-intersection})}
		\\&\ge \ttt_j - (1+2\epsilon) \cdot |N^G_j(v)| \tag{From (\ref{ineq:uw-approx})} 
	\end{align*}
	Thus, by inductive hypothesis, it holds that,
	\begin{align*}
		|N^G_j(\tilde{S} \cup \{v\})| &= |N^G_j(v)| + |N^{G'}(\tilde{S})|
		\\&\ge |N^G_j(v)| + (1-2\epsilon) \cdot \lr{\ttt_j - (1+2\epsilon) \cdot |N^G_j(v)|} - \frac{\epsilon (k-1)}{k^*} \cdot t_j
		\\&\ge (1-2\epsilon)\cdot \ttt_j + 4\epsilon^2 \cdot |N^G_j(v)| - \frac{\epsilon (k-1)}{k^*} \cdot t_j
		\\&\ge (1-2\epsilon) \cdot \ttt_j - \frac{\epsilon(k-1)}{k^*} \cdot t_j \tag{since $|N^G_j(v)| \ge 0$}
	\end{align*}
	This completes the induction, and thus the proof of the lemma.
\end{proof}

\begin{theorem} \label{thm:d-hs}
	There exists a randomized algorithm that runs in time $\lr{ \frac{6dr}{\epsilon} }^k \cdot \lr{\frac{18\log k}{\epsilon^2}}^{kr} \cdot n^{\Oh(1)}$, and given a \yin $\cI = (G = (A \uplus B, E), [r], f, t, k)$ of \pbds, where each vertex in $B$ has degree at most $d$, with probability at least $ \lr{ \frac{1}{L} \cdot \frac{\epsilon}{2rd} }^k$, returns a subset $\tilde{S} \subseteq A$ of size  $k$ such that $|N^G_j(\tilde{S})| \ge (1-\epsilon) t_j$ for all colors $j$. %; otherwise, if $\cI$ is a \nin, then the algorithm either returns no, or outputs a solution $\tilde{S}$ with the claimed properties.
\end{theorem}

Next, we conclude with the proof of Theorem~\ref{thm:d-hs}, 
%the following theorem, 
which follows from \Cref{lem:random-induction} in a straightforward manner.
%\begin{theorem} \label{thm:d-hs}
%	There exists a randomized algorithm that runs in time $\lr{ \frac{6dr}{\epsilon} }^k \cdot \lr{\frac{18\log k}{\epsilon^2}}^{kr} \cdot n^{\Oh(1)}$, and given a \yin $\cI = (G, [r], f, t, k)$ of \pbds, with high probability, returns a subset $\tilde{S} \subseteq A$ of size at most $k$ with coverage vector $(t''_1, t''_2, \ldots, t''_r)$, such that $t''_j \ge (1-\epsilon) t_j$ for all colors $j$; otherwise, if $\cI$ is a \nin, then the algorithm may detect that it is a \nin, or output a solution $\tilde{S}$ with the claimed properties.
%\end{theorem}
\begin{proof}[Proof of Theorem~\ref{thm:d-hs}]
	Let $S \subseteq A$ be a set of size $k$ such that $|N^G_j(S)| = \ttt_j \ge t_j$ for all colors $j$. Let $\tilde{S}$ denote the output of \textsc{PCCDS}$(\cI)$. It is easy to see that the algorithm returns a solution in polynomial time. Next, \Cref{lem:random-induction} implies that with probability at least $q = \lr{\frac{1}{L} \cdot \frac{\epsilon}{2rd}}^k$, the set $\tilde{S}$ satisfies that, for each $j \in [r]$,
	\begin{align*}
		|N^G_j(\tilde{S})| = t''_j \ge (1-2\epsilon) \cdot t_j - \frac{\epsilon k}{k^*} t_j = (1-2\epsilon) t_j - \epsilon t_j = (1-3\epsilon) \cdot t_j
	\end{align*} 
	Here, we use the fact that since $\cI$ is the original instance, we have $k^* = k$. Also note that, for any color $j \in [r], \ttt_j \ge t_j$.
	
	We make $\Oh(q^{-1} \log n)$ independent calls to PCCDS$(\cI)$, and if in any of the calls, we find a set $\tilde{S}$ with the claimed properties, then we return it. Otherwise, the algorithm concludes that $\cI$ is a \nin. We get the claimed running time by rescaling $\epsilon$ to $\epsilon/3$. 
\end{proof}

\section{An \FPTAS for \pbds on $K_{d,d}$-free graphs}
\label{sec:kdd-free}

In this section, we design an \FPTAS for \pbds on $K_{d,d}$-free graphs. % that runs in $2^{\Oh(\frac{k^3rd}{\epsilon})}(n+m)^{\Oh(1)}$ time. 
In the algorithm, we first divide the colors into two sets according to their coverage requirements: ${\sf T_{small}}\coloneqq \{j\in [r] \colon t_j \leq \nicefrac{2k^2d}{\epsilon}\}$ and ${\sf T_{large}}\coloneqq\{j\in [r] \colon t_j >  \nicefrac{2k^2d}{\epsilon}\}$. Further, we do bucketing of the vertices in ${\sf T_{large}}$ and ${\sf T_{small}}$ separately. For the vertices in ${\sf T_{large}}$, the strategy is similar to \textsc{Bucketing} in Section~\ref{subsec:bucketing}. For the sake of simplicity of analysis, we use $\Oh(\log m/\epsilon)$ buckets per color, instead of $\Oh(\log k/\epsilon^2)$ as in the previous section. As a result, we will get a slightly worse running time. Specifically, we will have an extra $\log r$ factor in the exponent. Note that this factor can be eliminated by using $\Oh(\log k/\epsilon^2)$ buckets and a more careful analysis similar to the previous section that keeps track of the additive errors for color $j$ incurred when we branch on a bucket that contains all the ``small-degree'' vertices of color $j$; however we omit this. 

We will use {\em color coding} to identify a solution that covers the required coverage for the colors in ${\sf T_{small}}$ with high probability. Thus, we first propose a randomized algorithm here, which will be derandomized later using the known tool of  {\em $(p,q)$-perfect hash family}~\cite{alon1995color,fomin2014efficient}. 

Henceforth, we will assume that we are given a \yin and show that the algorithm outputs an approximate solution with high probability--otherwise the algorithm will detect that we are given a \nin.
Hence, there exists a hypothetical solution $S$ such that for every $j\in [r]$, $|N_j(S)|\geq t_j$. Note that, for every $j\in {\sf T_{small}}$, $t_j \le |N_j(S)| \le \frac{2k^2d}{\epsilon}$. %Let $S_{\sf small}\subseteq S$ such that for every $j\in {\sf T_{small}}$, $|N_j(S_{\sf small})|\geq t_j$.  We first highlight all the vertices in $\cup_{j\in {\sf T_{small}}}N_j(S_{\sf small})$ using color coding. %Note that $|\cup_{j\in {\sf T_{small}}}N_j(S_{\sf small})|\leq \nicefrac{krd}{\epsilon}$. 
As a first step, we first use color coding in order to attempt to identify the vertices in each $B_j$, where $j\in T_{\sf small}$, that are covered by the solution. Without loss of generality, let ${\sf T_{small}}=\{1,\ldots,z\}$ and $B_{\sf small}=\cup_{j\in {\sf T_{small}}} B_j$. 

\begin{tcolorbox}[boxsep=5pt,left=5pt,top=5pt,colback=green!5!white,colframe=gray!75!black]
	{\bf  Separation of small cover:}  %For each $j\in [z]$, 
	Label the vertices of $B_{\sf small}$ uniformly and independently at random using $\frac{2k^2zd}{\epsilon}$ labels, say $1,\ldots,\frac{2k^2zd}{\epsilon}$. %$\frac{2(j-1)k^2d}{\epsilon}+1,\ldots,\frac{2jk^2d}{\epsilon}$. %, such that no two vertices in $S$ have same color. That is, each color from $\{1,\ldots,ns\}$ appears on at most one vertex of $S$.
\end{tcolorbox}

The goal of the labelling is that ``with high probability'', we label the vertices in $B_{\sf small}$ that are covered by the solution with distinct labels. Note that the solution can cover more than $t_j$ vertices of color $j$, however, we are only concerned with
$t_j$ vertices. The following proposition bounds the success probability.

\begin{proposition}{\rm \cite[Lemma 5.4]{DBLP:books/sp/CyganFKLMPPS15}}\label{prop:success-prob} Let $\cU$ be a universe and $X\subseteq \cU$. Let $\chi \colon \cU \rightarrow [|X|]$ be a function that colors  each element of $\cU$ with one of $|X|$ colors uniformly and independently at random. Then, the probability that the elements of $X$ are colored with pairwise distinct colors is at least $e^{-|X|}$.
\end{proposition}

For a vertex $v\in B_{\sf small}$, let $label(v)$ denote its label. For $X\subseteq B_{\sf small}$, $label(X)=\cup_{v\in X}label(v)$. Let ${\sf labels}=\{1,\ldots,\nicefrac{2k^2zd}{\epsilon}\}$. %We next guess the partition of ${\sf labels}=L_1\uplus \ldots \uplus L_{\ell}$. The goal of this guessing is that the vertices with same labels are covered using the same vertex.  
We next move to the bucketing step. We first create buckets with respect to all the colors in $j\in {\sf T_{large}}$. % using a procedure similar to the {\sc Bucketing} procedure in Section~\ref{sec:bucketing}. 

\medskip\noindent\bucketinglarge.
For every color $j \in {\sf T_{large}}$ and $1 \le \alpha \le \log_{(1+\epsilon)}m$, we define $\displaystyle A(j, \alpha) \coloneqq \LR{ v \in A \colon  \nicefrac{2kd}{\epsilon} \cdot (1+\epsilon)^{\alpha-1} < d_j(v) \leq \nicefrac{2kd}{\epsilon} \cdot (1+\epsilon)^\alpha  }.$ For all the smaller degrees, we have the following bucket.  $\displaystyle A(j, 0) \coloneqq \LR{ v \in A \colon d_j(v) \leq \nicefrac{2kd}{\epsilon}}.$ 

Next, we create buckets for all $j\in {\sf T_{small}}$ as follows.

\medskip\noindent\bucketingsmall.
%Let $\cI = (G = (A \uplus B, E), r, f, t, k)$ be the \emph{current} instance of \pbds. Then, 
For every $j \in {\sf T_{small}}$ and a set $\gamma \subseteq {\sf labels}$, we define $$A(\gamma) \coloneqq \LR{ v \in A\colon label(N(v)\cap B_{\sf small})=\gamma}$$ 
%We also define $A(j, 0) \coloneqq \LR{ v \in A : d_j(v) \ge t^*_j }$, and $A(j, \lambda+1) \coloneqq \LR{v \in A: d_j(v) < \frac{t^*_j}{(1+\epsilon)^{\lambda}} }$. 
\begin{sloppypar}
	We first create bags $A(\mathbf v)$ as defined in Section~\ref{subsec:bucketing}. In particular, let  $\mathbf{V} = \{0,1, \ldots, \log_{(1+\epsilon)}m\}^{r-z}$. Consider an arbitrary vector $\mathbf{v} \in \mathbf{V}$. Let $\mathbf{v} = (\alpha_{z+1}, \ldots, \alpha_{r})$. Then, $A(\mathbf{v}) \coloneqq \bigcap_{j = z+1}^{r} A(j, \alpha_j)$. For every  $\gamma \subseteq {\sf labels}$, $\mathbf{v}\in \mathbf{V}$, let 
	$A^\gamma(\mathbf v)=A(\mathbf v)\cap A(\gamma)$. %
	We call any such $A^\gamma(\mathbf{v})$  a \emph{bag}. For every  $\gamma \subseteq {\sf labels}$, we also add $A(\gamma)$ to our collection of bags. Thus, the number of bags is upper bounded by $2^{\frac{2k^2rd}{\epsilon}}(1+(\log_{(1+\epsilon)}m +1)^r) \le 2^{\frac{2k^2rd}{\epsilon}} \cdot r^{\Oh(r)} \cdot m^{\Oh(1)}$ via standard arguments. Note that these bags form a covering and not a partition of vertices in $A$ as was the case in Section~\ref{subsec:fptas-freq-d}
\end{sloppypar}

Our main idea is as follows. We start by guessing a bag that has a non-empty intersection with an optimal solution. Since every vertex in a bag is adjacent to vertices of the same label set, any vertex in the bag can be chosen in order to cover the vertices of colors $j\in T_{\sf small}$. Further, the $j$-degree of vertices in the same bag is ``almost'' equal, for every $j\in T_{\sf large}$.  We will demonstrate that selecting a vertex $v$ from a selected bag leads to one of the following two possibilities: either it belongs to the solution or, there exists at least one vertex from the set of vertices, each of whose neighborhood has significantly overlap with the $j$-neighborhood of $v$ for all $j\in T_{\sf large}$. The formal algorithmic description is presented in Algorithm~\ref{alg:kdd-free}.

To begin, we utilize the definition and lemma introduced by Jain et al.~\cite{DBLP:conf/soda/0001KPSS0U23} to elaborate on the concept of a ``high'' intersection. %We first define what we mean by vertices that have large $j$-neighborhood in $X\subseteq B$, which is the same as defined in. 
% high intersection set of a set of vertices $U$ and prove the following lemma.
\begin{definition}[$\beta_j$-High Degree Set]\label{def:def3}
	%\begin{cdef}[$\beta$-High Degree Set]
	Given a bipartite graph $G=(A,B, E)$, a set $X\subseteq B$, a color $j\in [r]$, and a positive integer $\beta > 1$, the {\em $\beta_j$-High Degree Set}, denoted by ${\sf HD}_{\beta_j}^{G}(X) \subseteq A$, is a set of vertices such that every vertex $v \in {\sf HD}_{\beta_j}^{G}(X) $ satisfies $|N_j(v)\cap X| \geq \frac{|X|}{\beta}$, i.e.,
	% and $|N(v)| > j$.
	\begin{center}
		$ {\sf HD}_{\beta_j}^{G}(X) = \{v\in A \colon |N_j(v)\cap X| \geq \frac{|X|}{\beta}\}$
	\end{center}
	
	%$(AHD_{\beta}(U))$
	
	%\kijfree
\end{definition}

Let ${\sf AHD}_{\beta_j}^{G}(X) = {\sf HD}_{\beta_j}^{G}(X) \cap \{v\in A: |N_j(v)| \geq d \}$. That is, ${\sf AHD}_{\beta_j}^{G}(X)$ consists of vertices of $j$-degree at least $d$ and those that belong to 
${\sf HD}_{\beta_j}^{G}(X)$.  Due to Lemma 4.2 in~\cite{DBLP:conf/soda/0001KPSS0U23}, we know that $|{\sf AHD}_{\beta_j}^{G}{(X)}| \leq (d-1)(2\beta)^{d-1}$, for all $X\subseteq B$, $d$, and for all $\beta > 1$ with $\frac{|X|}{2\beta} > d$.

\begin{algorithm}[t]
	\caption{\textsc{$K_{d,d}$-free-\pbds}$(\cI= (G = (A \uplus B, E), [r], f, t, k,\epsilon))$}
	\label{alg:kdd-free}
	\begin{algorithmic}[1]
		\Statex 
		\If{$k = 0$} 
		\State \Return $\emptyset$ if $t_j=0$, for every $j\in [r]$; otherwise {\sf NO}.
		\EndIf
		%\Statex {\small \texttt{All quantities in the following are computed w.r.t. $G$, but the subscript is omitted for simplicity.}}
		%\State let $\epsilon'=\frac{\epsilon}{k}$
		\State compute ${\sf T_{small}}$ and ${\sf T_{large}}$.
		\State label the vertices in $B_{\sf small}$ as defined above in the {\color{green}green} box. 
		%		\State let ${\cal P}$ be the set of partition of ${\sf labels}$.
		%	\For{every partition $P\in {\cal P}$}
		\State Let ${\cal A}$ be the set of bags returned by applying \bucketinglarge on ${\sf T_{large}}$ and \bucketingsmall on ${\sf T_{small}}$. %with respect to partition $P$
		\For{every bag $A^\gamma(\mathbf{v}) \in {\cal A}$}
		\State select an arbitrary vertex $x \in A^\gamma(\mathbf{v})$ \label{sep:arbitrary_vertex}
		\State compute ${\sf AHD}^G_{\beta_j}(N_j(x))$ for $\beta=\frac{k}{\epsilon}$, and $j\in {\sf T_{large}}$. % such that $\mathbf{v}(j)\neq 0$. 
		\State let $Z_{\mathbf{v}}= \bigcup\limits_{\substack{j\in {\sf T_{large}}}}{\sf AHD}^G_{\beta_j}(N(x)) \cup \{x\}$  
		\EndFor
		%	\EndFor
		\State let $Z=\bigcup\limits_{A^\gamma(\mathbf{v}) \in {\cal A}}Z_{\mathbf{v}}$
		\State from every bag $A(\gamma)\in {\cal A}$, add a vertex in $Z$\label{sep:vertex_small_cov} 
		\For{each $y\in Z$} 
		\State let $\cI_y \gets \textsc{PruneInstance}(\cI, y)$ 
		\State let $S_y$  be the set returned by \textsc{$K_{d,d}$-free-\pbds}($\cI_y$)   % and $A_y=A\setminus \{y\}$. 
		\EndFor
		\State Among all the sets $S_y$,  $y\in Z$, suppose $z$ is the element such that for every $j\in [r]$,
		$|N_j(\{z\}\cup S_z)|\geq (1-k\epsilon)t_j$. Return $\{z\}\cup S_z$. Return {\sf NO}, if no such $z$ exists.
		%\Statex \Comment{\small \texttt{$f'$ is a restriction of $f$ on the remaining vertices in $G\bbslash u$, and $\tau'$ is the \emph{residual coverage requirement}, i.e., $\tau'(j) = \max\LR{\tau(j) - |N_j(u)|, 0}$}}
	\end{algorithmic}
\end{algorithm}

%\begin{lemma}\label{lem:kdd-free-smallcolors}
%Given a yes-instance $(\cI= (G = (A \uplus B, E), [r], f, t, k,\epsilon))$ of \pbds where $G$ is a $K_{d,d}$-free graph, Algorithm~\ref{alg:kdd-free} finds a set $S\subseteq V(G)$ such that for every $j\in {\sf T_{small}}$, $|N_j(S)|\geq t_j$ with probability at least $e^{\nicefrac{-k^2dr}{\epsilon}}$; otherwise if $\cI$ is a no-instance, then the algorithm correctly concludes no. \todo{check this.}
%\end{lemma}

\begin{lemma}\label{lem:kdd-free}
	Given a \yin $(\cI= (G = (A \uplus B, E), [r], f, t, k,\epsilon))$ of \pbds where $G$ is a $K_{d,d}$-free graph, Algorithm~\ref{alg:kdd-free} finds a set $S\subseteq V(G)$ of size at most $k$ such that for every $j\in {\sf T_{large}}$, $|N_j(S)|\geq (1-k\epsilon)t_j$, and  for every $j\in {\sf T_{small}}$, $|N_j(S)|\geq t_j$   with probability at least $e^{\nicefrac{-2k^2k^\star dr}{\epsilon}}$. %; otherwise if $\cI$ is a \nin, then the algorithm either returns {\sf NO} or returns a set $S\subseteq V(G)$ such that for every $j\in [r]$, $|N_j(S)|\geq (1-k\epsilon)t_j$. %, for every $j\in [r]$.
\end{lemma}

\begin{proof}
	We prove it by induction on $k$. 
	
	\noindent {\em Base Case:} When $k=0$, then we cannot cover any vertex; thus the statement holds trivially.  
	
	\noindent{\em Induction Hypothesis:} Suppose that the claim is true for $k\leq \ell-1$. 
	
	\noindent {\em Inductive Step:} Next, we prove the claim for $k=\ell$. Let $S\subseteq V(G)$ such that $|S|=\ell$ and for every $j\in [r]$, $|N_j(S)|=t^\star_j \geq t_j$.  We consider the following two cases. 
	\begin{description}
		\item[Case 1.] $S \cap Z = \emptyset$.  Suppose that ${\sf T_{small}}=[r]$, i.e., for all $j\in [r]$, $t_j \leq \nicefrac{2\ell^2d}{\epsilon}$.  Then, $N(S)$ is colorful with probability at least $e^{\nicefrac{-2\ell^2d}{\epsilon}}$. Thus, $S$ has non-empty intersection with at least one bag $A(\gamma)$, where $\gamma \subseteq {\sf labels}$, with probability at least $e^{\nicefrac{-2\ell^2d}{\epsilon}}$. Let $x\in S\cap A(\gamma)$.  Let $p$ be an arbitrary vertex in $A(\gamma)$ selected in Step~\ref{sep:vertex_small_cov}  of the algorithm. Due to the construction of the bucket $A(\gamma)$, we know that $label(N(p)\cap B_{\sf small})=label(N(x)\cap B_{\sf small})=\gamma$.  Let $J\subseteq {\sf T_{small}}$ such that label of at least one vertex of $B_j$, $j\in J$, is in $\gamma$. %Thus, for every $j\in J$, $d_j(x)=d_j(p)$. 
		Note that $x$ and $p$ does not cover any vertex with color $j\in {\sf T_{small}}\setminus J$. %Since $|N_j(S')|\geq t_j-d_j^G(x)$, 
		Due to induction hypothesis, $|N_j(S_p)|\geq t_j-d_j^G(p)$ with probability at least $e^{\nicefrac{-2(\ell-1)^2k^\star dr}{\epsilon}}$, for every $j\in [r]$. Hence,  for every $j \in [r]$, $|N_j(\{p\}\cup S_p)| \geq t_j$ with probability at least $e^{\nicefrac{-2\ell^2k^\star dr}{\epsilon}}$.
		
		Next, we consider the case when ${\sf T_{large}}\neq \emptyset$. Then, for every $j\in {\sf T_{large}}$, there exists at least one vertex $v \in S$ such that $d_j(v)\geq \nicefrac{2\ell d}{\epsilon}$. Thus,  $S$ has non-empty intersection with at least one bag $A^\gamma(\mathbf v) \in {\cal A}$. Let $x\in S \cap A^\gamma(\mathbf v)$. Let $S'=S \setminus \{x\}$. Clearly, for every $j\in [r]$, $|N_j^G(S')|\geq t_j-d_j^G(x)$. %Without loss of generality, let ${\sf T_{small}}=\{1,\ldots,z\}$.  
		Let $p$ be an arbitrary vertex in $A^\gamma(\mathbf v)$ selected in Step~\ref{sep:arbitrary_vertex} of the algorithm.  Furthermore, note that $S\cap Z =\emptyset$. Thus, for every $j\in {\sf T_{large}}$ and $w\in S$,  either $N_j(w) < d$ or $|N_j(w)\cap N_j(p)| < \nicefrac{|N_j(p)|}{\beta}$. Thus, for every $j\in {\sf T_{large}}$,
		
		\begin{equation*}
			\begin{split}
				|N_j^{G_p}(S')| & = |N_j^{G}(S')\setminus N(p)| \\
				& \geq t_j-d_j^G(x) -\ell d - \ell \frac{|N_j(p)|}{\beta} \\
				& \geq t_j-d_j^G(x) -\ell d - \epsilon |N_j(p)|
			\end{split}
		\end{equation*}
		
		Due to induction hypothesis, for every $j\in {\sf T_{large}}$, $|N_j^{G_p}(S_p)| \geq (1- (\ell -1)\epsilon)(t_j-d_j^G(x) -\ell d - \epsilon |N_j^G(p)|)$. Next, we argue that for every $j\in {\sf T_{large}}$,
		$|N_j(\{p\} \cup S_p)|\geq (1-\ell \epsilon)t_j$. %If $\alpha_j=0$, then $p \in A(j, 0)$, it holds that $|N^G_j(p)| \ge t^*_j\ge t_j$. Thus, $|N^G_j(\{p\} \cup S_p)| \ge t_j$. Next, consider the case when $1\leq \alpha_j \leq \lambda$. In this case, 
		Since $p$ and $x$ belong to the same bag, $d_j(x)\leq (1+\epsilon)d_j(p)$. Note that 
		
		\begin{equation*}
			\begin{split}
				|N_j(\{p\} \cup S_p)| & = |N_j^G(\{p\}| + |N_j^{G_p}(S_p)| \\
				& \geq d_j^G(p) + (1-(\ell -1)\epsilon)(t_j-d_j^G(x) -\ell d - \epsilon d_j^G(p)) \\
				& \geq d_j^G(p)+ (1-(\ell -1)\epsilon)(t_j-(1+\epsilon)d_j(p)-\ell d-\epsilon d_j^G(p)) \\
				& \geq (1-(1-\epsilon)(1+2\epsilon'))d_j^G(p)+(1-\epsilon)(1-(\ell -1)\epsilon)t_j \\
				& \geq (1-\ell \epsilon)t_j 
				% & = (1-\epsilon)t_j 
			\end{split}
		\end{equation*}
		
		% Since $t_j \geq \frac{kd}{\epsilon}$, we have 
		%   \begin{equation*}
			%   \begin{split}
				%   |N_j(\{p\} \cup S_p)| & \geq  (1-\epsilon')^2t_j\\
				%   & \geq (1-2\epsilon')t_j
				%   \end{split}
			% \end{equation*}
		
		%Thus, for $\epsilon' = \frac{\epsilon}{2}$, we obtain that $|N_j(\{p\} \cup S_p)|\geq (1-\epsilon)t_j$. Next, we consider the case when $\alpha_j > \lambda$.
		%
		%
		%
		%
		%
		We next argue the claim for $j\in {\sf T_{small}}$. Using the same argument as above (when we considered ${\sf T_{small}}=[r]$), we have that for every $j\in {\sf T_{small}}$, $|N_j(\{p\}\cup S_p)| \geq t_j$ with probability at least $e^{\nicefrac{-2\ell^2k^\star dr}{\epsilon}}$. %$|N_j(S_p)|\geq (1-\epsilon)(t_j-d_j^G(p))$ with probability at least $e^{\nicefrac{-(k-1)kdr}{\epsilon}}$. Hence,  for $j\in {\sf T_{small}}$, with probability at least $e^{\nicefrac{-k^2dr}{\epsilon}}$
		%\begin{equation*}
		%\begin{split}
		%|N_j(\{p\}\cup S_p)| & =  |N_j(p)|+|N_j^{G_p}(S_p)| \\ 
		%& \geq d_j^G(p) + t_j-d_j^G(p) \\ 
		%& \geq t_j
		%\end{split}
		%\end{equation*}
		
		\item[Case 2] $S \cap Z \neq \emptyset$. Let $x\in S \cap Z$. Let $S' = S \setminus \{x\}$.  For every $j\in {\sf T_{large}}$, clearly, $|N_j(S')|\geq t_j - d_j^G(x)$, otherwise, $S$ is not a solution to $\cI$. Thus, $\cI_x$ is a \yin to \pbds, as $S'$ is a solution to $\cI_x$. 
		Hence, due to our induction hypothesis, we know that Algorithm~\ref{alg:kdd-free} finds a set $S_x \subseteq V(G_x)$ of size $\ell -1$ such that for every $j\in [r]$, $|N_j(S_x)|\geq (1-(\ell-1)\epsilon)(t_j-d_j^G(x))$. Thus, for every $j\in {\sf T_{large}}$,
		\[|N_j^G(\{x\}\cup S_x)|  = |N_j^G(x)| + |N_j^{G_x}(S_x)| \geq (1-\ell \epsilon)t_j\]
		
		For every $j\in {\sf T_{small}}$, the argument is same as in Case 1.
		Since there exists an element $x\in Z$ such that $|N_j^G(\{x\}\cup S_x)|\geq (1-\ell \epsilon)t_j$, our algorithm returns one such set. %\todo{add probability argument.}
		%Since $label(N(p)\cap B_{\sf small})=label(N(x)\cap B_{\sf small})=\gamma$, $x$ and $p$ does not cover any label that is not in $\gamma$.  Without loss of generality, let ${\sf T_{small}}=\{1,\ldots,\ell\}$. 
		%
		%Next, we argue for $j\in {\sf T_{large}}$. Let ${\mathbf v}= (\alpha_{z+1},\ldots,\alpha_r)$. We consider the following cases for $j\in \{z+1,\ldots,r\}$.
		%\begin{description}
		%\item[Case $\alpha_j >0$.] Since $p \in A(j, 0)$, it holds that $|N^G_j(p)| \ge t^*_j\ge t_j$. Thus, $|N^G_j(\{p\} \cup S_p)| \ge t_j$.
		%
		%\item[Case $0< \alpha_j < \lambda$.]	 Since $p,x \in A(j, \alpha_j)$, $d_j(p)\geq (1+\epsilon)d_j(x)$. 
		%
		%\end{description}
	\end{description}
\end{proof}

Thus, we obtain the following result by invoking Algorithm~\ref{alg:kdd-free} with $\epsilon'=\nicefrac{\epsilon}{k}$. 
\begin{theorem}\label{thm:kdd-free}
	There exists a randomized algorithm that runs in $2^{\Oh(\frac{k^4rd\log r}{\epsilon})}(n+m)^{\Oh(1)}$ time, and given a \yin $(\cI= (G = (A \uplus B, E), [r], f, t, k,\epsilon))$ of \pbds where $G$ is a $K_{d,d}$-free graph, finds a set $S\subseteq V(G)$ of size at most $k$ such that for every $j\in {\sf T_{large}}$, $|N_j(S)|\geq (1-\epsilon)t_j$, and  for every $j\in {\sf T_{small}}$, $|N_j(S)|\geq t_j$   with probability at least $e^{\nicefrac{-2k^3k^\star dr}{\epsilon}}$. %; otherwise if $\cI$ is a \nin, then the algorithm either returns {\sf NO} or returns a set $S\subseteq V(G)$ such that for every $j\in [r]$, $|N_j(S)|\geq (1-\epsilon)t_j$. %, for every $j\in [r]$.
\end{theorem}

%\begin{proof}
%We call  Algorithm~\ref{alg:kdd-free} with $\epsilon'=\nicefrac{\epsilon}{k}$. The correctness follows due to Lemma~\ref{lem:kdd-free}. Recall that the number of bags is upper bounded by $2^{\frac{2k^2rd}{\epsilon}}(1+(\log_{(1+\epsilon)}n +1)^r)$. Furthermore, for every $x\in A^{\gamma}({\mathbf v})$, $|{\sf AHD}_{\beta_j}^G(N(x))|\leq (d-1)(2\beta)^{d-1}$. Thus, in every recursive call,  $$|Z| \leq 2^{\frac{2k^2rd}{\epsilon}}\big(1+(d-1)(2\beta)^{d-1}(1+(\log_{(1+\epsilon)}n +1)^r)\big)$$
%
%Since the number of recursive calls is bounded by $k$, the running time is $2^{\Oh(\frac{k^3rd\log r}{\epsilon})}(n+m)^{\Oh(1)}$.
%\end{proof}

\begin{proof}
	We call Algorithm~\ref{alg:kdd-free} with $\epsilon'=\nicefrac{\epsilon}{k}$. The correctness follows due to Lemma~\ref{lem:kdd-free}. Recall that the number of bags is upper bounded by $2^{\frac{2k^3rd}{\epsilon}}(1+(\log_{(1+\epsilon)}m +1)^r)$. Furthermore, for every $x\in A^{\gamma}({\mathbf v})$, $|{\sf AHD}_{\beta_j}^G(N(x))|\leq (d-1)(2\beta)^{d-1}$. Thus, in every recursive call,  $$|Z| \leq 2^{\frac{2k^3rd}{\epsilon}}\big(1+(d-1)(2\beta)^{d-1}(1+(\log_{(1+\epsilon)}m +1)^r)\big)$$
	
	Since the number of recursive calls is bounded by $k$, the running time is $2^{\Oh(\frac{k^4rd\log r}{\epsilon})}(n+m)^{\Oh(1)}$.
\end{proof}

We derandomize this algorithm using $(p,q)$-perfect hash family to obtain a deterministic algorithm in the following theorem.
%Theorem~\ref{thm:introkdd}.
% for our problem. %Towards this, we first define the notion of $(p,q)$-perfect hash family. 

%We use $(p,q)$-perfect hash family to derandomize the algorithm in Theorem~\ref{thm:kdd-free}.

\begin{theorem}
	\label{thm:introkdd}
	There exists a deterministic algorithm that runs in $2^{\Oh(\frac{k^4r^2d\log r}{\epsilon})} \cdot (n+m)^{\Oh(1)}$ time, and given a \yin $(\cI= (G = (A \uplus B, E), [r], f, t, k,\epsilon))$ of \pbds where $G$ is a $K_{d,d}$-free graph, finds a set $S\subseteq A$ of size at most $k$ such that, for every $j\in [r]$, $|N_j(S)|\geq (1-\epsilon)t_j$. %; otherwise if $\cI$ is a \nin, then the algorithm either returns no or returns a set $S\subseteq V(G)$ such that for every $j\in [r]$, $|N_j(S)|\geq (1-\epsilon)t_j$.
	%There exists a deterministic algorithm that runs in $2^{\Oh(\frac{k^3r^2d\log r}{\epsilon})}\cdot (\log m)^k \cdot (n+m)^{\Oh(1)}$ time, and given a \yin $(\cI= (G = (A \uplus B, E), [r], f, t, k,\epsilon))$ of \pbds where $G$ is a $K_{d,d}$-free graph, finds a set $S\subseteq V(G)$ of size at most $k$ such that for every $j\in {\sf T_{large}}$, $|N_j(S)|\geq (1-\epsilon)t_j$, and  for every $j\in {\sf T_{small}}$, $|N_j(S)|\geq t_j$; otherwise if $\cI$ is a no-instance, then the algorithm either returns {\sf NO} or return a set $S\subseteq V(G)$ such that for every $j\in {\sf T_{large}}$, $|N_j(S)|\geq (1-\epsilon)t_j$, for every $j\in [r]$.
\end{theorem}

\begin{proof}[Proof of Theorem~\ref{thm:introkdd}]  
	\begin{definition}[$(p,q)$-perfect hash family]{\rm (\cite{alon1995color})}
		For non-negative integers $p$ and $q$, a family of functions $f_1,\ldots,f_t$ from a universe $\cU$ of size $p$ to a universe of size $q$ is called a $(p,q)$-perfect hash family, if for any subset $S\subseteq \cU$ of size at most $q$, there exists $i\in [t]$ such that $f_i$ is injective on $S$. 
	\end{definition}
	
	We can construct $(p,q)$-perfect hash family using the following result. 
	
	\begin{proposition}[\cite{NaorSS95,DBLP:books/sp/CyganFKLMPPS15}]\label{prop:hash family construction}
		There is an algorithm that given $p,q \geq 1$ constructs a $(p,q)$-perfect hash family of size $e^qq^{\Oh(\log q)}\log p$ in time $e^qq^{\Oh(\log q)}p \log p$.
	\end{proposition}
	
	Let $\cI$ be an instance of \pbds. Instead of taking a random coloring for $B_{\sf small}$ in Algorithm~\ref{alg:kdd-free}, we create a $(|B_{\sf small}|,\frac{2k^2rd}{\epsilon})$-perfect hash family ${\cal F}$, and run the algorithm for every label function $f\in {\cal F}$. Using this, we get the proof of Theorem~\ref{thm:introkdd}. 
\end{proof}

\section{Handling Matroid Constraints}\label{sec:matroidc}

Recall that we want to find a subset $A' \subseteq A$ that is independent in the given $\cM$ of rank at most $k$. Without loss of generality, we assume that $k$ equals the rank of $\cM$, i.e., the solution is a base of $\cM$ by truncating the matroid appropriately. Note that it is straightforward to work with the truncated matroid, given oracle access to the original matroid.
%This can be achieved in randomized polynomial time by truncating the matroid to size $k$ using the results of Marx \cite{DBLP:conf/acid/Marx06} (for matroids representable over special fields, this can in fact be done in deterministic polynomial time using the results of Lokshtanov et al.~\cite{DBLP:conf/icalp/LokshtanovMPS15}). 
By slightly abusing the notation, we use $\cM$ to denote the appropriately truncated matroid, if necessary.

\subsection{$K_{d,d}$-free Case} \label{subsec:kdd-matroid}
Now, we are ready to discuss our algorithm. Let $(\cI= (G = (A \uplus B, E), [r], f, t, k,\epsilon),\mathcal{M},S_{par})$ be an instance of \pmbds, where $\mathcal{M}=(A, I)$ is a matroid. The algorithm is largely similar to the one in Section~\ref{sec:kdd-free}, with a few modifications as described next. 
\medskip

\noindent\textbf{Modifications:} 
%Now we discuss the required  modifications to the algorithm in Section~\ref{alg:kdd-free} to make sure it works for \pmbds. 
In addition to the standard inputs for Algorithm~\ref{alg:kdd-free}, the modified algorithm instance also receives (oracle access to) a matroid $\mathcal{M'}=(A',I')$. Here, $\cM'$ is obtained by contracting the original matroid $\cM$ on the set of elements $Q$ added so far leading to this recursive call. From the definition of matroid contraction, it follows that any independent set in $\cM'$, along with $Q$, is independent in the original matroid $\cM$. Due to our initial truncation, we can inductively assume that $\cM'$ has rank exactly $k$, where $0 \le k \le k^{\star}$, where $k^{\star}$ is the \emph{original} budget (and thus the rank of the original matroid $\cM$).
%and a partial solution $S_{par}$ as inputs. 
Furthermore, rather than just searching for a set of size $k$ that satisfies the coverage requirements for each color class, the algorithm seeks a set $S$ that meets two conditions: first, it must be independent in $\mathcal{M'}$, and second, it must have a neighborhood size $N_j(S) \geq (1-\epsilon)t_j$, thus satisfying the coverage requirements for each color class $j \in [r]$. 

Similar to Algorithm~\ref{alg:kdd-free}, we start by guessing a bag $A^{\gamma}(\mathbf{v})$ that contains a vertex of an optimal solution $\mathsf{OPT}$ (i.e., we branch on all such bags). But instead of selecting any arbitrary vertex in  $A^{\gamma}(\mathbf{v})$ (as done in line 5 of Algorithm~\ref{alg:kdd-free}), we compute a $R(A^{\gamma}(\mathbf{v})) \rep{k-1} A^{\gamma}(\mathbf{v})$. \Cref{lem:oraclerepset} implies that, $|R(A^{\gamma}(\mathbf{v}))| \le k$, and it can be computed in polynomial time. Next, for every $v_i \in  R(A^{\gamma}(v))$, we compute the sets ${\sf AHD}_{\beta _j}^G (N(v_i))$. Next, we define 
\begin{equation}
	Z_{\mathbf{v}} \coloneqq \bigcup_{v_i \in R(A^{\gamma}(v)} \bigcup_{j \in \mathsf{T}_{large}} \{v_i\} \cup {\sf AHD}_{\beta _j}^G (N(v_i) \label{eqn:repset-z}
\end{equation}
%We now assert that at least one vertex $y \in Z \coloneqq \bigcup_{v_i\in R(A^{\gamma}(v)}\{v_i\} \cup {\sf AHD}_{\beta _j}^G (N(v_i))$ can be added to the gradually constructed solution  (returned by an iteration of the same algorithm with a smaller $k$) without violating independence while meeting the approximate coverage requirements. 
We branch on such a $y\in Z_{\mathbf{v}}$ and update the instance $\mathcal{I}'$ passed to the next iteration of the algorithm as was done in the $\mathsf{PruneInstance}$ procedure, with the following modification. 
%The two additional parameters of $\mathcal{M}$ and $S_par$ are modified as follows. First we update the partial solution to $S_{par}=S_{par}\cup \{y\}$. 
First, we obtain $\cM'' \coloneqq \cM' / y$, i.e., $\cM''$ is obtained by contracting $\cM'$ on the vertex $y$ on which we are branching. Note that one can simulate oracle access to $\cM''$ using the oracle access to $\cM'$ by always including $y$ in the set being queried.
%Note that given the linear representation of $\cM'$, one can easily obtain the linear representation of $\cM''$ in polynomial time (see, e.g., \cite{DBLP:conf/acid/Marx06}). 
Note that the rank of $\cM''$ is $k-1$. 
%Next, we delete all the vertices from the universe $A$ that may be spanned by the set $S_{par}\cup \{y\}$ (the newly to be a solution) and update the matroid $\mathcal{M}$ by deleting all these vertices from its universe as well as independent sets. This operation is also well-known as the \emph{contraction operation} $\mathcal{M}/(S_{par}\cup \{y\})$. We do an additional update on $G'$ by deleting all such spanned vertices.  
For the sake of formality, we describe the explicit changes made to the algorithm.
\medskip

\noindent \textbf{Exact changes}: 
\begin{itemize}
	\item Line 8 of Algorithm~\ref{alg:kdd-free} is replaced by: Compute $R(A^{\gamma}(\mathbf{v})) \rep{k-1} A^{\gamma}(\mathbf{v})$ using Lemma \ref{lem:oraclerepset}. 
	\item In Line 10 of Algorithm~\ref{alg:kdd-free}, the set $Z_v$ is defined as in (\ref{eqn:repset-z}). 
	
	\item In $\mathsf{PruneInstance}$ ($\mathcal{I},y$), we also compute the contracted matroid $\cM''$ by contracting $\cM'$ on the element $y$ as mentioned above. 
\end{itemize}
\medskip

\noindent\textbf{Correctness:} We sketch the modified algorithm's correctness through induction on the parameter $k$. This approach is similar to how we proved the correctness of Algorithm~\ref{alg:kdd-free}. For the base case where $k = 0$, correctness is trivial. Assuming that the algorithm correctly returns an approximate solution when $k \leq i$, we will prove the correctness for the case of $k = i+1$. For that purpose we consider the following two scenarios assuming the input to be a \textsf{Yes}-instance:

\begin{itemize}
	\item Case 1: ($\textsf{OPT} \cap Z \neq \emptyset$) \
	Let $x\in Z \cap \textsf{OPT}$. By induction our algorithm on the instance returned by $\mathsf{PruneInstance}$ ($\mathcal{I},x$) with the contracted matroid $\mathcal{M''} = \cM/x$, returns a set of size $S'$ of size $k-1$ that satisfies the approximate coverage requirements (as argued for  Algorithm~\ref{alg:kdd-free}). By the definition of $\cM''$, it follows that $S' \cup \LR{x}$ is independent in $\cM'$.
	\iffalse
	This is valid because the matroid $\mathcal{M}$ passed to the next iteration of the function is contracted (deleting every element that can be in the span of the elements in $Z\cap\textsf{OPT}$ ). 
	\fi
	\item Case 2: ($\textsf{OPT} \cap Z = \emptyset$) \
	Recall that $A^{\gamma}(\mathbf{v}) \cap \textsf{OPT} \neq \emptyset$, i.e., $\OPT$ selects at least one vertex, say $x$ from $A^{\gamma}(\mathbf{v})$. However, $\OPT \cap Z = \emptyset$, which implies that $x \not\in Z$, which, in particular, implies that $x \not\in R(A^{\gamma}(\mathbf{v}))$. In this case, based on the correctness arguments of Algorithm~\ref{alg:kdd-free}, we know that any vertex $y \in A^{\gamma}(\mathbf{v}) \cap Z$ serves as a suitable ``approximate replacement'' for $x$, as far as the coverage requirement is concerned. 
	
	However, here we have an additional requirement that that the solution be an independent set in $\cM'$. To this end, let $\OPT' = \OPT \setminus \LR{x}$. Note that $|\OPT'| = k-1$, and $\OPT' \cup \LR{x}$ is independent in $\cM'$. It follows that, there exists some $y \in R(A^{\gamma}(\mathbf{v})) \rep{k-1} A^{\gamma}(\mathbf{v})$ such that $y \cap \OPT' = \emptyset$ and $y \cup \OPT'$ also independent in $\cM'$. Thus, using inductive hypothesis, the solution returned by the recursive call corresponding to $y$, combined with $y$, is (1) independent in $\cM'$, and (2) satisfies the coverage requirements up to an $1-\epsilon$ factor.
\end{itemize}

\iffalse
This confirms the existence of a vertex $u$ in the $k$-representable set, such that the solution returned by the pruned instance (when $y$ is included) combines with $S'$ to form an approximate solution that is also independent,
thereby demonstrating the correctness of the inductive step.
\fi

\noindent\textbf{Running time:}
Note that the branching factor in line 10 of Algorithm~\ref{alg:kdd-free} increases by at most $k$. This adds a multiplicative factor of $k^k$ to the running time, which is absorbed into the FPT factor. Furthermore, the time required to compute a representative set and $Z_{\mathbf{v}}$ being polynomial-time for any bag $A^{\gamma}(v)$ is absorbed into the polynomial factor. Thus, we obtain the following result.

\begin{theorem}\label{thm:kdd-free}
	There exists a randomized algorithm that runs in $2^{\Oh(\frac{k^3rd\log r}{\epsilon})}(n+m)^{\Oh(1)}$ time, and given a yes-instance $\cI= (G = (A \uplus B, E), [r], f, t, k,\epsilon,\mathcal{M})$ of \mfmaxcov where $G$ is a $K_{d,d}$-free graph, finds a set $S\subseteq V(G)$ of size at most $k$ that is independent in $\mathcal{M}$ such that for every $j\in {\sf T_{large}}$, $|N_j(S)|\geq (1-\epsilon)t_j$, and  for every $j\in {\sf T_{small}}$, $|N_j(S)|\geq t_j$   with probability at least $e^{\nicefrac{-2k^2k^\star dr}{\epsilon}}$; otherwise if $\cI$ is a no-instance, then the algorithm either returns {\sf NO} or return a set $S\subseteq V(G)$ such that for every $j\in {\sf T_{large}}$, $|N_j(S)|\geq (1-\epsilon)t_j$, for every $j\in [r]$.
\end{theorem}

\subsection{Frequency $d$ Case}

\noindent \textbf{Modifications:} In addition to the standard inputs for Algorithm~\ref{alg:freqd}, the modified
algorithm instance is provided a matroid $\mathcal{M}'$ and a partial solution $S_{par}$ as inputs. The algorithm seeks an independent set $S$ in $\mathcal{M}'$ that satisfies the coverage requirements approximately. Similar to Algorithm~\ref{alg:freqd}, we start by guessing a bag $A(\mathbf{v})$ that contains a vertex of an optimal solution \textsf{OPT}. But instead of selecting any arbitrary vertex in $A(\mathbf{v})$ (as done in line
5 of Algorithm 2), we compute a $R(A(\mathbf{v}))\rep{k-1} A(\mathbf{v})$ of size at most $k$ in polynomial-time. We choose a vertex $v$ uniformly at random from this set and subsequently proceed in accordance with the steps outlined in Algorithm~\ref{alg:freqd}. We contend that with a \emph{high} probability, either vertex $v$ or a vertex $w$ (based on the probability distribution $p(w)$) can be included in the gradually constructed solution (produced by an iteration of the same algorithm with smaller $k$) without compromising independence, while still satisfying the approximate coverage requirements. And, the \textsf{PruneInstance} procedure undergoes identical modifications as detailed in the preceding section.

\noindent\textbf{Exact Changes:}
\begin{itemize}
	\item Line 5 of Algorithm~\ref{alg:freqd} is replaced by: Choose a bag $A(\mathbf{v}) \in \mathcal{A}$ uniformly at random. Compute $R(A(\mathbf{v}))\rep{k-1} A(\mathbf{v})$ using Lemma~\ref{lem:oraclerepset}. Uniformly at random select a vertex $v$ from it.
	\item In $\mathsf{PruneInstance}$ ($\mathcal{I},y$): the matroid $\cM'$ passed to the new instance is obtained by contracting $\cM$ on $y$, i.e., $\cM' = \cM/y$.
	%Add the following line between lines 10 and 11- Delete $\mathsf{S_P}$ which contains all $v\in A$ that are spanned by $S_{par}\cup\{y\}$ and in $\mathcal{I}'$, update the matroid $\mathcal{M}'=\mathcal{M}[U\setminus \mathsf{S_P}]$ \todo {how to write this new matroid by deleting certain vertices} and $S_{par}'=S_{par}\cup \{y\}$. Also update $G'=G'-\mathsf{S_P}$.
\end{itemize}
\medskip

\noindent\textbf{Correctness:} We establish the correctness of the algorithm by reasoning that, with a sufficiently high probability at each step, we either choose a vertex from the optimal solution (\textsf{OPT}) or select a vertex that can be added to the resulting solution $\tilde{S}$ while preserving independence and satisfying the approximate coverage requirements. Notice that while in Algorithm  selecting an arbitrary vertex was sufficient, that may not remain true in the presence of a matroid constraint. Since one may not be able to add such a vertex while keeping the set $(\tilde{S}\cup\{u\})$ independent. Hence we compute $R(A(\mathbf{v}))$ that contains at least one vertex that may be added to the returned $\tilde{S}$ while preserving independence. Note that the probability of selecting such a vertex from $R(A(\mathbf{v}))$ is at least $1/k$, and it worsens the success probability by the same factor. We denote an optimal solution by both $S$ and \textsf{OPT} to maintain consistency with the notations from the previous section.
\begin{itemize}
	\item  Case 1: $R(A(\mathbf{v})) \cap S \neq \emptyset$ \\
	In this case, a vertex is chosen $R(A(v)) \cap \textsf{OPT}$ into the solution with a probability of at least $1/2k$.
	\medskip
	
	\item  Case 2: $R(A(\mathbf{v})) \cap S= \emptyset$\\
	If there is a color $j$ and  a $v\in R(A(\mathbf{v}))$ such that $\sum_{w' \in S} |N^G_j(v) \cap N^G_j(w')| \ge \epsilon \cdot |N^G_j(v)|$, then the probability that some vertex $w$ from the set $S$  in line 10 is at least $(1/k)\frac{\epsilon}{r}$. The rest of the argument follows similar to the arguments in \Cref{lem:random-induction} but with a probability worsened by a factor of $1/k$. Otherwise, for all colors $j$, and all vertices $v\in R(A(v))$ it holds that $\sum_{w' \in S} |N_j(v) \cap N_j(w')| \le \epsilon \cdot |N_j(v)|$. But, in this case, it was shown \Cref{lem:random-induction} that $v$ ``approximately plays the role'' of $w \in A(\mathbf{v}) \cap S$ when there is no matroid-constraint. In the matroid-constraint case, we can show that there exists a vertex $v \in R(A(v))$ that not only satisfies the approximate coverage requirements with $\tilde{S}$ but also forms an independent set. And, the probability that such a vertex $v$ is chosen in the branching step is $1/2k$ (probability worsens by a factor of $1/k$).
\end{itemize}
\medskip

\noindent\textbf{Running time:}
Note that the probability of a "good event" in the modified algorithm deteriorates by a maximum factor of $1/k$ at each branching step. This introduces an additional run time of $k^k$. And, the additional time taken to compute a representative family  being polynomial-time for any bag $A^{\gamma}(\mathbf{v})$ is absorbed into the polynomial factor of the algorithm's run time.

\begin{theorem} \label{thm:d-hs-matroid}
	There exists a randomized algorithm that runs in time $ \lr{ \frac{2kdr}{\epsilon} }^k \cdot \lr{\frac{6\log k}{\epsilon^2}}^{kr} \cdot n^{\Oh(1)}$, and given a yes-instance $\cI = (G, [r], f, t, k, \mathcal{M})$ of \mfmaxcov, where each element appears in at most $d$ sets, with high probability, returns a subset $\tilde{S} \subseteq A$ of size at most $k$ with coverage vector $(t''_1, t''_2, \ldots, t''_r)$, such that $t''_j \ge (1-3\epsilon) t_j$ for all colors $j$; otherwise, if $\cI$ is a no-instance, then the algorithm correctly concludes so.
\end{theorem}

%The algorithm can be derandomized using $(p,q)$- perfect hash family to obtain a deterministic algorithm in Theorem \red{blah}.

\paragraph{Extension to Intersection of Multiple Linear Matroids.} The above approach can be extended to the more general problem where the solution is required to be an independent set in each of the given matroids $\cM_1, \cM_2, \ldots, \cM_q$ that are all defined over the common ground set $A$, \emph{if all the given matroids are representable over a field $\mathbf{F}$}. The only change here is that, the representative family $R(A^\gamma(\mathbf{v}))$ needs to be computed for the direct sum of the matroids $\cM'_1 \oplus \cM'_2 \oplus \cdots \oplus \cM'_q$, where $\cM'_i$ denotes the contracted version of the $i$-th matroid as defined above. Then, one can use the linear algebraic tools to compute a representative set of the direct sum matroid in a specific manner. For more details, we refer the reader to Marx \cite{DBLP:conf/acid/Marx06}, Section 5.1. Due to this change, the bound on $R(A^{\gamma}(\mathbf{v}))$ becomes $qk$ (from $k$), and the time required to compute this set is at most $2^{\Oh(qk)} \cdot (m+n)^{\Oh(1)}$. This also gets reflected in the running time of the algorithm.

\section{Conclusion}
In this paper, we designed {\sf FPT}-approximation schemes for \mfmaxsat, which is a generalization of the \satcc problem with fairness and matroid constraints. %The problem is inherently {\sf W}[2]-hard with respect to $k$, and also no hope to design an algorithm that runs in $h(k)(n+m)^{o(k)}$ time and outputs a sub-family of size $k$ that covers at least $(1-\frac{1}{e} +\epsilon)t$ elements, even for $r=1$. Thus, inspired from the recent works for the {\sc Maximum Coverage} problem, we considered set systems satisfying additional properties. 
%In particular, we designed {\sf FPT-AS} for $K_{d,d}$-free set systems and a faster algorithm for set systems with bounded frequency.  An interesting direction of research is to design lossy kernels for the considered problem. It would also be interesting to generalize these results for {\sc Partition Max-SAT} where, given a  CNF-formula $\Phi$ such that the clauses are partitioned into $r$ color classes with coverage demand $t_j$ for every color $j\in [r]$, the goal is to find a minimum weight assignment that satisfies at least $t_j$ clauses in every color class $j\in [r]$.  
In particular, we designed {\sf FPT-AS} for the classes of formulas where the maximum frequency of a variable in the clause is bounded by $d$, and more generally, for $K_{d,d}$-free formulas. Our algorithm for \fmaxcov on the set systems of frequency bounded by $d$ is substantially faster compared to the recent result of Bandyapadhyay et al. \cite{BandyapadhyayFM23}, even for the special case of $d = 2$. We use a novel combination of the bucketing trick and a carefully designed probability distribution in order to obtain this faster \FPTAS. 

Our work naturally leads to the following intriguing questions. Firstly, our approximation-preserving reduction from \satcc (and variants) to \maxcov (and variants) is inherently randomized. Is it possible to derandomize this reduction? A similar question of derandomization is also interesting for our aforementioned algorithm for \fmaxcov on bounded-frequency set systems. In this case, can we design an \FPTAS for the problem running in time single-exponential in $k$? %Finally, our algorithms for the matroid constrained problems rely on efficient computation of representative sets, and thus only work for linearly representable matroids (unlike \cite{Sellier23}). Is it possible to extend the \FPTAS{}es to handle arbitrary matroid constraints? 
%Finally, can we design \FPT approximation schemes for more general class of set systems satisfying certain natural structural properties? %In this direction, we have a preliminary evidence suggesting that one can obtain \FPTAS for \mfmaxsat on $K_{d,d}$-free set systems, we leave the details for a future version.

%Finally, we also obtained an algorithm that approximates the solution size $k$. In particular, we design an \FPT approximation algorithm, parameterized by $k+r$, that given a  \yin, returns a sub-family of size $k+r$ that covers $t_j$ elements for every color class $j$. This approximation algorithm is asymptotically tight, that is, there is no \FPT algorithm, paramterized by $k+r$, that can produce a sub-family of size $k+r-1$, unless $\FPT=\mathsf{W}[1]$. 

\paragraph{Acknowledgments.} We thank Petr Golovach and an anonymous reviewer for \Cref{lem:oraclerepset}. 
%%
%% Bibliography
%%
%% Please use bibtex, 
%\newpage

\bibliographystyle{siam}
\bibliography{referencesA}
%\bibliography{references,referencesA}

%\appendix 
%\input{appendix}
%\bibliographystyle{siam}
%\bibliography{references}

%\appendix

\end{document}

%% file: preamble.tex
\usepackage{amsmath,array,bm}
\usepackage{amsmath,amsfonts}
\usepackage{amsthm}
\usepackage{soul}
\usepackage{amssymb,latexsym}
\usepackage[ruled]{algorithm}
\usepackage{subcaption}
\usepackage{algorithmicx}
\usepackage[]{algpseudocode}
\usepackage{bbm}
\usepackage{mathrsfs}
\usepackage{graphicx}
\usepackage{caption}
\usepackage{enumitem}
\usepackage{eufrak}
\usepackage{float}
\usepackage{thmtools}
\usepackage{hyperref}
\usepackage{mathtools}
\usepackage{cleveref}
\usepackage{bussproofs}
\usepackage[square,sort,comma,numbers]{natbib}
\usepackage{tikz}
\usepackage{stmaryrd}
\usepackage{xspace}
\usepackage{xcolor}
\usepackage{marginnote}
\usepackage{tipa}
\usepackage{mdframed}
\usepackage[export]{adjustbox}[2011/08/13]
\usepackage{multirow}
\usepackage{relsize}
\usepackage{hhline}
\usepackage{todonotes}
 \usepackage{tcolorbox} 
\usepackage{nicefrac}

\crefname{invar}{invariant}{invariants}
\crefname{ineq}{inequality}{inequalities}
\crefname{constr}{constraint}{constraints}
\crefname{tbl}{table}{tables}
\crefname{lem}{lemma}{lemmata}
\crefname{lemma}{lemma}{lemmata}
\crefname{cond}{condition}{conditions}
\newcommand{\red}[1]{{\color{red} #1}}
\newcommand{\blue}[1]{{\color{blue} #1}}
\newcommand{\lr}[1]{\left( #1\right)}
\newcommand{\LR}[1]{\left\{ #1\right\}}

\newcommand{\Oh}{\mathcal{O}}

\newcommand{\cU}{\mathcal{U}}

\newcommand{\cH}{\mathcal{H}}
\newcommand{\cI}{\mathcal{I}}
\newcommand{\cF}{\mathcal{F}}

\newcommand{\bluecomment}[1]{\Comment{ \blue{\small\texttt{#1}}}}

\newcommand{\fs}{\text{\large $\mathfrak{s}$}}

\newcommand{\ttt}{\widetilde{t}}

\newcommand{\FPT}{\textsf{\upshape FPT}\xspace}
\newcommand{\FPTAS}{\textsf{\upshape FPT-AS}\xspace}

\newcommand{\OPT}{\textsf{OPT}}

\newcommand{\cnf}{\textsf{CNF}}
\newcommand{\var}[1]{{\sf var}(#1)}
\newcommand{\cla}[1]{{\sf cla}(#1)}

\newcommand{\pbdsfull}{\textsc{Partition Color Constrained Dominating Set}\xspace}
\newcommand{\pbds}{\textsc{PCCDS}\xspace}

\newcommand{\pmbds}{\textsc{PMCCDS}\xspace}

\newcommand{\bucketing}{\textsc{Bucketing}\xspace}
\newcommand{\bucketingsmall}{\textsc{Bucketing(small)}\xspace}
\newcommand{\bucketinglarge}{\textsc{Bucketing(large)}\xspace}

\newcommand{\cB}{\mathcal{B}}

\newtheorem{theorem}{Theorem}[section]
\newtheorem{lemma}[theorem]{Lemma}
\newtheorem{claim}[theorem]{Claim}
\newtheorem{corollary}[theorem]{Corollary}
\newtheorem{definition}[theorem]{Definition}
\newtheorem{observation}[theorem]{Observation}
\newtheorem{proposition}[theorem]{Proposition}

\newtheorem{remark}[theorem]{Remark}

\newcommand{\mmaxcov}{\textsc{M-MaxCov}\xspace}
\newcommand{\mmaxsat}{\textsc{M-MaxSAT}\xspace}
\newcommand{\fmaxcov}{\textsc{F-MaxCov}\xspace}
\newcommand{\fmaxsat}{\textsc{F-MaxSAT}\xspace}
\newcommand{\mfmaxcov}{\textsc{(M, F)-MaxCov}\xspace}
\newcommand{\mfmaxsat}{\textsc{(M, F)-MaxSAT}\xspace}

%% file: Multicolored_Coverage_long.bbl
\begin{thebibliography}{10}

\bibitem{alon1995color}
{\sc N.~Alon, R.~Yuster, and U.~Zwick}, {\em Color-coding}, JACM, 42 (1995),
  pp.~844--856.

\bibitem{AminiFS11}
{\sc O.~Amini, F.~V. Fomin, and S.~Saurabh}, {\em Implicit branching and
  parameterized partial cover problems}, J. Comput. Syst. Sci., 77 (2011),
  pp.~1159--1171.

\bibitem{DBLP:journals/corr/abs-2308-15842}
{\sc S.~Bandyapadhyay, A.~Banik, and S.~Bhore}, {\em On colorful vertex and
  edge cover problems}, Algorithmica,  (2023), pp.~1--12.

\bibitem{BandyapadhyayFM23}
{\sc S.~Bandyapadhyay, Z.~Friggstad, and R.~Mousavi}, {\em A parameterized
  approximation scheme for generalized partial vertex cover}, in {WADS} 2023,
  P.~Morin and S.~Suri, eds., 2023, pp.~93--105.

\bibitem{DBLP:journals/tcs/BeraG0R14}
{\sc S.~K. Bera, S.~Gupta, A.~Kumar, and S.~Roy}, {\em Approximation algorithms
  for the partition vertex cover problem}, Theor. Comput. Sci., 555 (2014),
  pp.~2--8.

\bibitem{DBLP:journals/ipl/Blaser03}
{\sc M.~Bl{\"{a}}ser}, {\em Computing small partial coverings}, Inf. Process.
  Lett., 85 (2003), pp.~327--331.

\bibitem{CalinescuCPV11}
{\sc G.~C{\u{a}}linescu, C.~Chekuri, M.~P{\'{a}}l, and J.~Vondr{\'{a}}k}, {\em
  Maximizing a monotone submodular function subject to a matroid constraint},
  {SIAM} J. Comput., 40 (2011), pp.~1740--1766.

\bibitem{DBLP:journals/jco/ChekuriIQVZ22}
{\sc C.~Chekuri, T.~Inamdar, K.~Quanrud, K.~R. Varadarajan, and Z.~Zhang}, {\em
  Algorithms for covering multiple submodular constraints and applications}, J.
  Comb. Optim., 44 (2022), pp.~979--1010.

\bibitem{DBLP:conf/icalp/Cohen-AddadG0LL19}
{\sc V.~Cohen{-}Addad, A.~Gupta, A.~Kumar, E.~Lee, and J.~Li}, {\em Tight {FPT}
  approximations for k-median and k-means}, in {ICALP} 2019, C.~Baier,
  I.~Chatzigiannakis, P.~Flocchini, and S.~Leonardi, eds., 2019,
  pp.~42:1--42:14.

\bibitem{DBLP:books/sp/CyganFKLMPPS15}
{\sc M.~Cygan, F.~Fomin, L.~Kowalik, D.~Lokshtanov, D.~Marx, M.~Pilipczuk,
  M.~Pilipczuk, and S.~Saurabh}, {\em Parameterized Algorithms}, Springer,
  2015.

\bibitem{DBLP:conf/stoc/Feige96}
{\sc U.~Feige}, {\em A threshold of ln \emph{n} for approximating set cover
  (preliminary version)}, in {STOCS}, 1996, G.~L. Miller, ed., {ACM}, 1996,
  pp.~314--318.

\bibitem{DBLP:journals/algorithms/FeldmannSLM20}
{\sc A.~E. Feldmann, K.~{C. S.}, E.~Lee, and P.~Manurangsi}, {\em A survey on
  approximation in parameterized complexity: Hardness and algorithms},
  Algorithms, 13 (2020), p.~146.

\bibitem{FominLRS11}
{\sc F.~V. Fomin, D.~Lokshtanov, V.~Raman, and S.~Saurabh}, {\em Subexponential
  algorithms for partial cover problems}, Inf. Process. Lett., 111 (2011),
  pp.~814--818.

\bibitem{fomin2014efficient}
{\sc F.~V. Fomin, D.~Lokshtanov, and S.~Saurabh}, {\em Efficient computation of
  representative sets with applications in parameterized and exact algorithms},
  in proceedings of {SODA}, 2014, pp.~142--151.

\bibitem{GuoNW07}
{\sc J.~Guo, R.~Niedermeier, and S.~Wernicke}, {\em Parameterized complexity of
  vertex cover variants}, Theory Comput. Syst., 41 (2007), pp.~501--520.

\bibitem{DBLP:journals/tcs/HuangS23}
{\sc C.~Huang and F.~Sellier}, {\em Matroid-constrained maximum vertex cover:
  Approximate kernels and streaming algorithms}, in {SWAT} 2022, A.~Czumaj and
  Q.~Xin, eds., 2022, pp.~27:1--27:15.

\bibitem{DBLP:journals/algorithmica/HungK22}
{\sc E.~Hung and M.~Kao}, {\em Approximation algorithm for vertex cover with
  multiple covering constraints}, Algorithmica, 84 (2022), pp.~1--12.

\bibitem{DBLP:conf/soda/0001KPSS0U23}
{\sc P.~Jain, L.~Kanesh, F.~Panolan, S.~Saha, A.~Sahu, S.~Saurabh, and
  A.~Upasana}, {\em Parameterized approximation scheme for biclique-free max
  \emph{k}-weight {SAT} and max coverage}, in Proceedings of the 2023
  {ACM-SIAM} Symposium on Discrete Algorithms, {SODA} 2023, Florence, Italy,
  January 22-25, 2023, N.~Bansal and V.~Nagarajan, eds., {SIAM}, 2023,
  pp.~3713--3733.

\bibitem{KoanaKNS22}
{\sc T.~Koana, C.~Komusiewicz, A.~Nichterlein, and F.~Sommer}, {\em Covering
  many (or few) edges with k vertices in sparse graphs}, in {STACS} 2022,
  P.~Berenbrink and B.~Monmege, eds., 2022, pp.~42:1--42:18.

\bibitem{DBLP:conf/icalp/LokshtanovMPS15}
{\sc D.~Lokshtanov, P.~Misra, F.~Panolan, and S.~Saurabh}, {\em Deterministic
  truncation of linear matroids}, in {ICALP} 2015, M.~M. Halld{\'{o}}rsson,
  K.~Iwama, N.~Kobayashi, and B.~Speckmann, eds., 2015, pp.~922--934.

\bibitem{LokshtanovPRS17}
{\sc D.~Lokshtanov, F.~Panolan, M.~S. Ramanujan, and S.~Saurabh}, {\em Lossy
  kernelization}, in {STOC} 2017, H.~Hatami, P.~McKenzie, and V.~King, eds.,
  {ACM}, 2017, pp.~224--237.

\bibitem{DBLP:journals/corr/abs-1810-03792}
{\sc P.~Manurangsi}, {\em A note on max k-vertex cover: Faster fpt-as, smaller
  approximate kernel and improved approximation}, in {SOSA} 2019, J.~T. Fineman
  and M.~Mitzenmacher, eds., 2019, pp.~15:1--15:21.

\bibitem{DBLP:conf/acid/Marx06}
{\sc D.~Marx}, {\em A parameterized view on matroid optimization problems}, in
  {ACiD} 2006, H.~Broersma, S.~S. Dantchev, M.~Johnson, and S.~Szeider, eds.,
  vol.~7 of Texts in Algorithmics, King's College, London, 2006, p.~158.

\bibitem{DBLP:journals/cj/Marx08}
\leavevmode\vrule height 2pt depth -1.6pt width 23pt, {\em Parameterized
  complexity and approximation algorithms}, Comput. J., 51 (2008), pp.~60--78.

\bibitem{NaorSS95}
{\sc M.~Naor, L.~J. Schulman, and A.~Srinivasan}, {\em Splitters and
  near-optimal derandomization}, in 36th Annual Symposium on Foundations of
  Computer Science, Milwaukee, Wisconsin, USA, 23-25 October 1995, {IEEE}
  Computer Society, 1995, pp.~182--191.

\bibitem{DBLP:conf/esa/Sellier23}
{\sc F.~Sellier}, {\em Parameterized matroid-constrained maximum coverage}, in
  {ESA} 2023, I.~L. G{\o}rtz, M.~Farach{-}Colton, S.~J. Puglisi, and G.~Herman,
  eds., vol.~274 of LIPIcs, 2023, pp.~94:1--94:16.

\bibitem{DBLP:journals/jair/SkowronF17}
{\sc P.~Skowron and P.~Faliszewski}, {\em Chamberlin-courant rule with approval
  ballots: Approximating the maxcover problem with bounded frequencies in {FPT}
  time}, J. Artif. Intell. Res., 60 (2017), pp.~687--716.

\bibitem{DBLP:journals/algorithmica/Sviridenko01}
{\sc M.~Sviridenko}, {\em Best possible approximation algorithm for {MAX} {SAT}
  with cardinality constraint}, Algorithmica, 30 (2001), pp.~398--405.

\bibitem{DBLP:journals/combinatorica/Wolsey82}
{\sc L.~A. Wolsey}, {\em An analysis of the greedy algorithm for the submodular
  set covering problem}, Comb., 2 (1982), pp.~385--393.

\end{thebibliography}
